\documentclass[10pt,aps,prd,onecolumn,eqsecnum,amsmath,amssymb,nofootinbib,superscriptaddress]{revtex4-2}

\usepackage{mathtools}			
\usepackage{amsthm}
\usepackage{float}
\restylefloat{table}
\usepackage[usenames]{color}
\usepackage{xcolor}
\usepackage{enumitem}

\usepackage[colorlinks=true,
            citecolor=red,
            linkcolor=blue,
            urlcolor=violet,
            filecolor=cyan,
            backref=false]{hyperref}

\setlistdepth{9}
\renewlist{itemize}{itemize}{9}

\usepackage{pifont}

\setlist[itemize,1]{label=\textbullet}
\setlist[itemize,2]{label=\textasteriskcentered}
\setlist[itemize,3]{label=--}
\setlist[itemize,4]{label=\textbullet}
\setlist[itemize,5]{label=-}
\setlist[itemize,6]{label=\textbullet}
\setlist[itemize,7]{label=-}
\setlist[itemize,8]{label=\textbullet}
\setlist[itemize,9]{label=-}

\newcommand{\dd}{{\rm{d}}} 
\newcommand{\bl}{\mbox{\boldmath$\ell$}}

\newcommand{\rovno}{& = &}

\def \rr {\rho}

\def \rn {R_0}
\def \BE {\begin{equation}}
\def \EE {\end{equation}}
\def \BEA { \begin{eqnarray}}
\def \EEA {\end{eqnarray}}
\def \bea { \begin{eqnarray}}
\def \eea {\end{eqnarray}}
\def \be {\begin{equation}}
\def \ee {\end{equation}}

\def \TK {\mathcal{\tilde K}}
\def \KK {\mathcal{ K}}

\def \pul {\textstyle{\frac{1}{2}}}
\def \ctvrt {\textstyle{\frac{1}{4}}}
\def \tre {\textstyle{\frac{1}{3}}}
\def \sest {\textstyle{\frac{1}{6}}}
\def \H {\mathcal{H}}

\def \BoxK {\triangle}

\newtheorem{theorem}{Theorem}[section] 
 
\newtheorem{lemma}[theorem]{Lemma}

\begin{document}

\title{Static spherically symmetric Kundt vacuum solutions of higher-derivative gravities}

\author{Breno L. Giacchini}
\email{breno.giacchini@matfyz.cuni.cz}
\affiliation{Institute of Theoretical Physics, Faculty of Mathematics and Physics, Charles University, V Hole\v{s}ovi\v{c}k\'ach 2, Prague 180 00, Czech Republic}

\author{Ivan Kol\'a\v{r}}
\email{ivan.kolar@matfyz.cuni.cz}
\affiliation{Institute of Theoretical Physics, Faculty of Mathematics and Physics, Charles University, V Hole\v{s}ovi\v{c}k\'ach 2, Prague 180 00, Czech Republic}

\author{Vojt\v{e}ch Pravda}
\email{pravda@math.cas.cz}
\affiliation{Institute of Mathematics of the Czech Academy of Sciences, \v{Z}itn\'a 25, 115 67 Prague 1, Czech Republic}

\author{Alena Pravdov\'a}
\email{pravdova@math.cas.cz}
\affiliation{Institute of Mathematics of the Czech Academy of Sciences, \v{Z}itn\'a 25, 115 67 Prague 1, Czech Republic}

\date{\today}

\begin{abstract}

We study static spherically symmetric Kundt solutions to the vacuum field equations of quadratic gravity with a cosmological constant, as well as specific models of six-derivative gravity. In quadratic gravity, we identify all solutions for coupling constants satisfying ${\alpha\neq3\beta}$, while the case ${\alpha=3\beta}$ is studied using the Frobenius method, where we derive the recurrence relations for the power series. In contrast, in six-derivative gravity, we focus on selected models to illustrate the variety of closed-form solutions; we also analyze possible indicial families of Frobenius solutions. For all solutions, we analyze curvature singularities and their accessibility to geodesic observers. 

We then construct exact gravitational-wave solutions propagating on some of these backgrounds in quadratic and six-derivative gravity. It is known that in Einstein gravity, gravitational waves on the Nariai background unavoidably contain singularities, which are interpreted as physical sources generating these gravitational waves. In contrast, in addition to singular solutions,  for appropriate values of the coupling constants, higher-order gravities allow for globally smooth solutions representing gravitational waves.

\end{abstract}

\maketitle

\tableofcontents

\section{Introduction}
		
Birkhoff’s theorem is a classical result in general relativity. It states that any spherically symmetric vacuum solution to Einstein's equations
is locally isometric to a region in the Schwarzschild spacetime. It can be generalized to the case with a non-zero cosmological constant $\Lambda$. However, for $\Lambda>0$, two branches of solutions appear (see \cite{Schleich2010,Morrow-Jones93}) - apart from Schwarzschild-de Sitter, there is also another static, spherically symmetric vacuum solution: Nariai spacetime $dS_2 \times S^2$, belonging to the Kundt\footnote{Kundt spacetimes admit a null geodesic, non-expanding, shearfree, and twistfree congruence.} class of spacetimes.

A natural question is to what extent 	Birkhoff’s theorem and the uniqueness of static spherically symmetric vacuum spacetimes generalize to more general metric theories of gravity governed by the Lagrangian
\be
{ L}={ L}(g_{ab},R_{abcd},\nabla_{a_1}R_{bcde},\dots,\nabla_{a_1\dots a_p}R_{bcde}), \label{Lagr}
\ee
where $L$ is a polynomial curvature invariant.

Quadratic gravity represents an important  
theory within the class of theories \eqref{Lagr} in which all vacuum solutions of general relativity (from now on, possibly including the cosmological constant $\Lambda$) are also vacuum solutions. 
Thus, the Schwarzschild-(A)dS and Nariai spacetimes are also vacuum solutions of quadratic gravity. However, the uniqueness of the Schwarzschild solution in quadratic gravity is lost. 
Remarkably, it has been shown that in quadratic gravity, another static, spherically symmetric, asymptotically flat black hole exists \cite{Luetal15}.
Furthermore, \cite{PraPraPodSva21} found that additional static, spherically symmetric Kundt spacetimes appear as vacuum solutions of quadratic gravity besides the Nariai spacetime. Thus,  the uniqueness of the Nariai spacetime within static, spherically symmetric Kundt spacetimes is also lost in quadratic gravity.

In the generic case, the Schwarzschild metric does not appear as a vacuum solution to higher-order gravity theories \eqref{Lagr} beyond quadratic gravity. This is in contrast to the special properties of the Nariai spacetime. 
The Nariai spacetime belongs \cite{Herviketal15,HerPraPra17} to the class of universal spacetimes \cite{HerPraPra14}  and thus will appear as a vacuum solution of any theory \eqref{Lagr}. For universal spacetimes, the field equations of higher-order gravities \eqref{Lagr} reduce to a single algebraic condition for $\Lambda$.
In particular, for the Nariai spacetime, solving the field equations of higher-order gravity \eqref{Lagr} amounts to solving an algebraic constraint relating the radius $\rn$ of the $S^2$ sphere appearing in the Nariai metric, the cosmological constant $\Lambda$ and other parameters of the model.

For example, in  the presence of cubic curvature corrections in \eqref{Lagr}, $\Lambda$ is shifted from its Einstein (and quadratic gravity) value of $1/\rn^2$ to
\be
\Lambda= \frac{1}{\rn^2}-\frac{8 (18 \eta_3+6 \eta_7+\eta_8)}{9 \gamma  \rn^6},
\ee
where the quantities $\eta_i$ and $\gamma$ in the last term correspond to coupling constants (see \eqref{action-6der} below). Remarkably, in this case, the Nariai spacetime can be a vacuum solution even if $\Lambda=0$~\cite{Giacchini:2025gzw}.

Besides the Nariai spacetime, which appears in all theories of the form \eqref{Lagr}, there may be additional static spherically symmetric Kundt spacetimes solving the vacuum field equations of a particular higher-order gravity.\footnote{In addition to the above-mentioned examples of such Kundt spacetimes in quadratic gravity \cite{PraPraPodSva21}, a six-derivative gravity example of such a Kundt spacetime is given in \cite{Giacchini:2025gzw}.} This will obviously depend on the particular choice of the theory, and possibly also on the values of the coupling constants and the cosmological constant.

In this paper, we systematically study static, spherically symmetric Kundt spacetimes in quadratic gravity. We also identify several such closed-form Kundt metrics as vacua of various six-derivative gravities. However, the complexity of six-derivative gravity does not allow for a systematic study of these solutions.

In addition, after identifying static and spherically symmetric Kundt backgrounds in quadratic and six-derivative gravity, we study gravitational waves propagating on these backgrounds. A well-known example in the context of general relativity is the Khlebnikov-Ghanam-Thompson spacetime \cite{Khlebnikov86}. It represents gravitational waves propagating on the Nariai background. This spacetime is conjectured to be universal \cite{HerPraPra14} (see also \cite{HerPraPra17}), and indeed, it naturally also solves the vacuum equations of quadratic and six-derivative gravity.

In general relativity, vacuum gravitational waves on the Nariai background unavoidably contain singularities. These singularities are interpreted as physical point-like sources or topological defects generating these gravitational waves \cite{Ortaggio02}.
In contrast, we show that in addition to singular solutions, higher-order gravities allow for globally smooth solutions representing gravitational waves for appropriate values of the coupling constants.

Naturally, we also discuss gravitational waves propagating on other static, spherically symmetric Kundt backgrounds that appear in addition to the Nariai spacetime in quadratic and six-derivative gravity theories.

\hspace{3mm}

The paper is structured as follows: In Sec.~\ref{sc:Kundtmetric}, we introduce the metric ansatz for a static, spherically symmetric Kundt spacetime and briefly discuss spacetimes within this class that are constructed as the direct product of two constant-curvature spaces. We also examine  the main conditions related to the divergence and regularity of scalar curvature invariants  and the timelike and null-geodesic completeness.

In Sec.~\ref{sc:QG}, we substitute this ansatz into the field equations of quadratic gravity  and find all solutions for coupling constants satisfying ${\alpha\neq3\beta}$; in the case ${\alpha=3\beta}$, besides constructing particular closed-form solutions, we analyze the Frobenius series solutions and derive the recurrence relations for their coefficients.

Sec.~\ref{sc:SDG} is devoted to six-derivative gravity, where we focus on selected models admitting closed-form solutions. In addition,  we  identify the indicial families of the Frobenius solutions. Some of these solutions to higher-derivative theories are then used in Sec.~\ref{sc:gravwaves} as backgrounds to construct exact gravitational-wave solutions.

The paper concludes with a summary of our results in Sec.~\ref{sc:summary}. Appendix~\ref{sc:Pboxeq} contains supplemental material on solving the polynomial Laplacian, Poisson, and heat equations on the 2-sphere, Appendix~\ref{sec_con} briefly discusses the connection between exact and power-series solutions in six-derivative gravity, and Appendix~\ref{AppC} establishes the triviality of the standard static, spherically symmetric Kundt sector (showing that only maximally symmetric spacetimes arise).

\section{Static spherically symmetric Kundt metric }
\label{sc:Kundtmetric}

Kundt metrics are metrics admitting a twist-free, shear-free, and non-expanding null geodesic congruence. The existence of such a congruence leads to a substantial simplification of Einstein's equations, and Kundt spacetimes have been thoroughly studied in the context of general relativity (see \cite{Stephanibook,GriPodbook}). In addition, various Kundt spacetimes have also been studied in the context of quadratic gravity and other modified theories of gravity.

In this work, we focus on Kundt spacetimes that are static and possess spherical symmetry. The usual form of a static, spherically symmetric metric reads
\be
	\dd s^2=-f(r) \dd t^2+\frac{1}{g(r)} \dd  r^2+ r^2
    ( \dd \theta^2+\sin^2\theta\dd\phi^2). \label{SSexpandingmetric}
	\ee
    This is a Petrov type D (or O) metric in which both principal null directions (PNDs) are expanding. In Kundt spacetimes, the Kundt null congruence is always a (non-expanding) PND. Thus, Petrov type D Kundt spacetimes are incompatible with the metric \eqref{SSexpandingmetric}.  For  Petrov type O, it can be shown  (see Appendix \ref{AppC}) that the only Kundt spacetimes compatible with \eqref{SSexpandingmetric}  are (A)dS and Minkowski spacetimes.

   Thus metric \eqref{SSexpandingmetric}  is not suitable for studying static spherically symmetric Kundt spacetimes.

In fact,  static and spherically symmetric Kundt metrics (apart from (A)dS and Minkowski spacetimes) belong to a lesser-known second class of static, spherically symmetric metrics, derived, e.g., in \cite{Guilabert24}, 
	\be
	\dd s^2=-f(r) \dd t^2+\frac{1}{f(r)} \dd  r^2+ \rn^2
    ( \dd \theta^2+\sin^2\theta\dd\phi^2), \label{SSKundtmetric}
	\ee
where $\rn$ is constant. For this  Petrov type D (or O) metric, the principal null congruence is a Kundt congruence.
    The transformation of this metric to the Kundt coordinates is given in Sec. \ref{sc:gravwaves}. 

Note that the form of the metric \eqref{SSKundtmetric} is preserved under the transformation
\bea
\bar t &= & \mu t,\label{tr_t}\\
\bar r &=  & \frac{r-\nu}{\mu}\label{tr_r},
\eea
with
\be
\bar f = \mu^{-2} f.
\ee
One can use \eqref{tr_t}, \eqref{tr_r} to reduce the number of free parameters in a given solution by at most two.

Metric \eqref{SSKundtmetric}  is a direct product metric of two two-dimensional metrics, the 2-dimensional Lorentzian spacetime \newline $\dd {s_{(I)}}^2 = -f(r) \dd t^2+\frac{1}{f(r)} \dd r^2$ and the homogeneous metric on the 2-sphere, with the Ricci scalars
\be
R_{(I)} = -f'',  \ \ \ \   R_{(II)} = \frac{2}{\rn^2},\label{curv-2spaces}
\ee
respectively.
Thus, the Ricci scalar of the resulting four-dimensional spacetime reads
\BE
    R = \frac{2}{{\rn}^2} - f''. \label{eqricci}
\EE
It is also useful to express the boost-weight zero component of the traceless Ricci tensor $S_{01} \equiv S_{ab} \ell^a n^b$ in the null frame\footnote{The null frame reads $\ell=\frac{1}{\sqrt{2}} (\sqrt{f} \dd t + \frac{1}{\sqrt{f}} \dd r )$, \ $n=\frac{1}{\sqrt{2}} (\sqrt{f} \dd t - \frac{1}{\sqrt{f}} \dd r )$,\  $ m^{(2)}= \rn \dd \theta $, \ $m^{(3)}= \rn \sin \theta \dd \phi.$  }
\BE
S_{01} = \frac{1}{4} \left( \frac{2}{{\rn}^2} + f''  \right) .\label{eqS01}
\EE

For the case with $f^{(3)}=0$, the metric \eqref{SSKundtmetric} becomes a direct product of two spaces of constant curvature. Such spacetimes have been studied in the context of general relativity (see \cite{Stephanibook,GriPodbook} and references therein) and include the Nariai (vacuum), Bertotti-Robinson (electrovacuum), and Pleba\'nski-Hacyan (electrovacuum) spacetimes. 

In quadratic gravity, some of the above electrovacuum spacetimes become vacuum spacetimes. Furthermore, in some cases, we will also obtain quadratic-gravity generalizations of these spacetimes.

To get a preliminary overview of the $f^{(3)}=0$ cases in quadratic gravity, let us include a summary in Table \ref{tableNariaietal}.  Note, however, that quadratic gravity also includes $f^{(3)} \neq 0$ solutions that are not covered by the table. The derivation of the results in this table can be found in Sec. \ref{sc:QG}.

\vspace{2mm}
\begin{small}
\begin{table}[h]
		\begin{center}
			\begin{tabular}{|c|c|c|c|c|c|}
				\hline
$R_{(I)}$ & geometry& constraints  & $R$ & $  S_{01}$ & spacetime \\[1mm] \hline\hline
$+$ & $dS_2\times S^2$ & & $+$  & $\neq 0$& Bachian-Nariai\\
$+$ & $dS_2\times S^2$ & $R_{(I)}=R_{(II)}$ & $+$ & $0$ & Nariai\\[1mm] \hline
$0$ & $M_2\times S^2$  & & $+$ & $+$ & Pleba\'nski-Hacyan\\[1mm] \hline 
$-$ & $AdS_2\times S^2$ & $|R_{(I)}|<R_{(II)}$
& $+$  & $+$ & Bachian-Bertotti-Robinson\\
$-$ & $AdS_2\times S^2$ & $|R_{(I)}|>R_{(II)}$ & $-$  & $+$ & Bachian-Bertotti-Robinson\\ 
$-$ & $AdS_2\times S^2$& $R_{(I)}=-R_{(II)}$ & $0$ & $+$ & Bertotti-Robinson\\
\hline
\end{tabular}\\[2mm]
\end{center}
	\caption{Spacetimes occurring as a direct product of two constant-curvature spaces compatible with the metric \eqref{SSKundtmetric} in general relativity and quadratic gravity. The first column contains the sign of Ricci scalar of the first block, $R_{(I)}$. Also note that due to \eqref{eqInvars}, $R=0$ implies vanishing of $C_{abcd} C^{abcd}$. This is consistent with   Bertotti-Robinson being of Petrov type O and remaining spacetimes of Petrov type D. }
\label{tableNariaietal}
\end{table}
\end{small}


	 For metric \eqref{SSKundtmetric}, the Einstein equations with cosmological constant $\Lambda$ reduce to
    \be
   f'' + 2 \Lambda = 0, \ \ \ \Lambda=\frac{1}{\rn^2},
    \ee
    which gives the  Nariai spacetime
    \be
	 f( r)=	- \Lambda   r^2  + b  r + c,\ \ \Lambda=\frac{1}{\rn^2}, \label{Nariai}
	\ee
    the only vacuum spacetime in general relativity within the class of metrics \eqref{SSKundtmetric}. As indicated in Table \ref{tableNariaietal}, quadratic gravity allows for more vacuum solutions of the form \eqref{SSKundtmetric}.

Following \cite{Frolov:2021afd}, one actually can deduce that any polynomial scalar curvature invariant  of   metric \eqref{SSKundtmetric} and the Riemann tensor (without its derivatives) is given as a polynomial of the two invariants (essentially the Ricci scalars of the two two-dimensional metrics):
\begin{equation}
    p=\frac{1}{{\rn}^2}, \quad v=-\frac{f''}{2}.
\end{equation}
For example,
\begin{equation}
\begin{aligned}
    R&=2(p+v)&&=\frac{2}{{\rn}^2} - f'',
    \\ 
    R_{ab}R^{ab}&=(p-v)^2+(p+v)^2&&=\frac{2}{\rn^4}+\frac{f''^2}{2},
    \\
    R_{abcd}R^{abcd}&=4(p^2+v^2)&&=\frac{4}{\rn^4}+f''^2, 
    \\
    C_{abcd}C^{abcd}&=\frac{4}{3}(p+v)^2&&=\frac{(-2+\rn^2 f'')^2}{3\rn^4}.
    \label{eqInvars}
\end{aligned}
\end{equation}
As a result, the regularity of any polynomial curvature invariant is equivalent to the regularity of $f''$.

Consider general geodesics, which can always be taken to lie in the equatorial plane ${\theta = \pi/2}$ without loss of generality, with the tangent vector $k^\mu = (\dot{t}, \dot{r}, 0, \dot{\phi})$. Putting together the normalization condition, $k^{\mu}k_{\mu}=-f\dot{t}^2+\dot{r}^2/f+\rn^2\dot{\phi}^2=\epsilon$, ${\epsilon=\pm1,0}$, the energy conservation, ${-k_{\mu}(\partial/\partial t)^{\mu}=f\dot{t}=E}$, and angular momentum conservation, ${k_{\mu}(\partial/\partial \phi)^{\mu}=\rn^2\dot{\phi}=L}$, we obtain ${\dot{r}=\pm \sqrt{E^2+f(r)(\epsilon -L^2/\rn^2)} }$. This, upon inverting and integrating for the affine parameter, leads to 
\begin{equation}\label{eq:affpar}
    \lambda-\lambda_0=\pm\int^{r}_{{r_0}}\frac{dr'}{\sqrt{E^2+f(r')\left(\epsilon -\frac{L^2}{\rn^2}\right)}}.
\end{equation}
Clearly, $r$ is directly the affine parameter of null radial geodesics (${\epsilon=0}$, ${L=0}$) irrespective of $f$. The convergence of the integral in \eqref{eq:affpar} can be used to test whether a curvature singularity, if present, is actually reached by timelike ${(\epsilon=-1)}$ or null ${(\epsilon=0)}$ geodesics, or if it is essentially harmless.

\section{Quadratic gravity }
\label{sc:QG}

	In this section, we identify all spacetimes of the form \eqref{SSKundtmetric} beyond the Nariai spacetime—the only vacuum general relativity solution compatible with \eqref{SSKundtmetric}—that solve the vacuum field equations of quadratic gravity with coupling constants satisfying $\alpha\neq 3\beta$. Also, for the particular case $\alpha=3\beta$, we obtain solutions expandable as Frobenius series. Some of these spacetimes are already listed in Table \ref{tableNariaietal} to highlight their connection to known  (vacuum and electrovacuum) solutions in general relativity. 

    \vspace{10mm}

	\subsection{Field equations}
	The action of quadratic gravity reads
	\be
	S = \int \dd^4 x\, \sqrt{-g}\, \Big(
	\gamma \,(R-2\Lambda) +\beta\,R^2  - \alpha\, C_{abcd}\, C^{abcd}	\Big)\,,
	\label{action}
	\ee
	with ${\gamma=1/G}$ ($G$ is the Newtonian constant), $\Lambda $ is the cosmological constant, and  $\alpha$, $\beta$ are coupling constants of quadratic gravity. 
    
    Corresponding  field equations of quadratic gravity read
	\be
	\gamma \left(R_{ab} - {\pul} R\, g_{ab}+\Lambda\,g_{ab}\right)-4 \alpha\,B_{ab}
	+2\beta\left(R_{ab}-\tfrac{1}{4}R\, g_{ab}+ g_{ab}\, \Box - \nabla_b \nabla_a\right) R = 0 \,, \label{EqQG4d}
	\ee
	where $B_{ab}$ is the {Bach tensor}
	\BE
	B_{ab} \equiv \big( \nabla^c \nabla^d + {\pul} R^{cd} \big) C_{acbd} \,, \label{defBach}
	\EE
	or equivalently
	\BE
	B_{ab}={\pul}\Box R_{ab} -{\sest}\big( \nabla_a \nabla_b +{\pul} g_{ab}\Box \big)R
	-{\tre}RR_{ab}+R_{acbd}\,R^{cd}+{\ctvrt}\big({\tre}R^2-R_{cd}R^{cd}\big)g_{ab} \,.\label{BachRicci}
	\EE
The	Bach tensor is traceless, symmetric, conserved, and well behaved under a conformal transformation $g_{ab}=\Omega^2 \tilde g_{ab}$:
	\begin{equation}
		g^{ab}B_{ab}=0 \,, \qquad B_{ab}=B_{ba} \,, \qquad
		\nabla^b B_{ab}=0
		\,, \qquad B_{ab}=\Omega^{-2}\tilde B_{ab}\,.
		\label{Bachproperties}
	\end{equation}

    It follows from \eqref{BachRicci} that the Bach tensor vanishes for Einstein spacetimes. Consequently,   Einstein spacetimes identically solve the vacuum field equations of quadratic gravity \eqref{EqQG4d} (see \cite{Buchdahl48}). 
    Furthermore, the last equation in \eqref{Bachproperties} implies that the Bach tensor vanishes for conformal-to-Einstein spacetimes as well. However, the Bach tensor can also vanish for spacetimes that are not conformal to Einstein (see \cite{Liu2013} for examples). 
	
Note that the quadratic-gravity field equations reduce considerably for spacetimes with a constant Ricci scalar $R=4\Lambda$:
\BE
R_{ab}-\Lambda \, g_{ab}=4k\, B_{ab}\,,\qquad \hbox{where} \qquad
k \equiv \frac{\alpha}{\gamma+8\beta\Lambda} ,
\label{fieldeqsEWmod}
\EE
assuming $\gamma+8\beta\Lambda\not=0$.
However, in this paper, we will also encounter cases with non-constant $R$ and, thus, we will work with the more general form \eqref{EqQG4d} of the field equations.

For the metric \eqref{SSKundtmetric}, the Bach tensor possesses two independent components
\BEA
B_{tt} &=& \frac{1}{24\,\rn^4}\; f\left\{4 + \rn^4\left[-\left(f''\right)^2 + 2\,f'\,f^{(3)} + 4\,f\,f^{(4)}\right]\right\}, \\
B_{rr} &=& \frac{1}{24\,\rn^4\,f}\; \left\{-4 + \rn^4\left[\left(f''\right)^2 - 2\,f'\,f^{(3)}\right]\right\}.
\EEA
A linear combination of these two components yields
\BE
\frac{B_{tt}}{f} + f\,B_{rr} = \frac{1}{6}\; f\,f^{(4)}. \label{Bachcond}
\EE
Consequently, $f^{(4)} = 0$ is a necessary condition for the vanishing of the Bach tensor. In fact, the Bach tensor vanishes if and only if:
\BE
f=a r^3 + b r^2 +c r + d, \quad \text{where} \quad b^2-3 a c = \frac{1}{\rn^4}. \label{Bachvanishescond}
\EE

For metric \eqref{SSKundtmetric}, the field equations of quadratic gravity \eqref{EqQG4d} reduce to 
\begin{small}
\bea
&& E_{tt}:  {\rn}^4 (\alpha -3 \beta ) \left(-4 f f^{(4)}+f''^2-2
f^{(3)} f'\right)-6 \gamma  \Lambda  {\rn}^4+6 \gamma  {\rn}^2 -4 \alpha +12 \beta =0, \label{FEQEtt} \\
&& E_{rr}: {\rn}^4 (\alpha -3 \beta ) \left(f''^2-2 f^{(3)}
f'\right)-6 \gamma  \Lambda  {\rn}^4+6 \gamma  {\rn}^2 - 4 \alpha +12 \beta = 0, \\
&& E_{\theta \theta}: {\rn}^4 \left[ (\alpha -3 \beta ) f''^2+3 \gamma  f''-2
(\alpha +6 \beta ) \left(f f^{(4)}+f^{(3)} f'\right)\right]+6 \gamma 
\Lambda  {\rn}^4 -4 \alpha +12 \beta = 0 .\label{eqEthth}
\eea
\end{small}
From the generalized Bianchi identity $\nabla^b E_{ab} = 0$ it follows that the $tt$ and $rr$ components are not independent, as
\be
\frac{( E^r{}_{r} - E^t{}_{t} ) f'}{2f} + E^r{}_{r}{}^\prime = 0 .
\label{eq:Biattrr}
\ee
Therefore, whenever $E_{rr}=0$ is satisfied, then also $E_{tt}=0$.
Similarly, $E_{\theta\theta},_\theta+\cot{\theta}
\left(
E_{\theta\theta}-E_{\phi\phi} \frac{1}{\sin^2{\theta}}\right)=0$ and  when $E_{\theta\theta}=0$ then also $E_{\phi\phi}=0$. Thus, in fact, there are only two independent equations to be solved, $E_{rr}=0$ and $E_{\theta\theta}=0$.

The trace of the field equations reads
\be
4 \gamma  \Lambda-\frac{2 \gamma }{{\rn}^2} -6 \beta  f f^{(4)}+\gamma  f''-6 \beta  f^{(3)}  f' = 0.
\label{eq:tr315}
\ee

It is also particularly useful to consider the combination 
$E_{tt}-E_{rr}$, which gives
\be
\rn (\alpha- 3 \beta) f f^{(4)} =0. \label{eqQGcomb}
\ee

From $\eqref{eqQGcomb}$, it is clear that two branches of solutions, $\alpha = 3\beta$ and $f^{(4)} = 0$, arise. Let us first focus on the case  $\alpha \neq 3\beta$, where $f$ is at most a third-degree polynomial in $r$. We will treat the cases for the various degrees of $f$ separately.

Before we proceed, note that our definition of the quadratic gravity action~\eqref{action} is such that if $\gamma=0$ then there is no cosmological constant term. Although this is not the most general form of a quadratic gravity action, regarding static spherically symmetric Kundt solutions it actually represents no loss of generality. In fact, defining a density of cosmological constant $\rho_\Lambda$ that is independent of $\gamma$, i.e., $\Lambda=\rho_\Lambda/\gamma$, it can be shown that if a model with $\gamma=0$ admits such a Kundt solution, then it must also have $\rho_\Lambda=0$. To prove this statement, according to Eq.~\eqref{eqQGcomb}, it is sufficient to consider the two scenarios, $\alpha=3\beta$ and $f^{(4)}=0$. In the first, taking $\gamma\to 0$ in Eq.~\eqref{FEQEtt} immediately gives $\rho_\Lambda=0$. In the second, taking $\gamma\to 0$ in Eq.~\eqref{eq:tr315} yields $3\beta  f^{(3)}  f'=2\rho_\Lambda$, hence the only nontrivial case to be considered is when $f$ is a cubic polynomial. As it turns out (see Sec.~\ref{sec_a3} below), this case requires $\beta=0$, implying $\rho_\Lambda=0$.

\subsection{Case $\alpha \neq  3 \beta$ and $f=a_0 + a_1 r + a_2 r^2 + a_3 r^3$ with $a_3 \neq  0$ }

\label{sec_a3}

In this case, Eq. \eqref{eqEthth} implies 
\be
\beta=0=\gamma
\ee
and quadratic gravity reduces to Weyl conformal gravity with field equations being $B_{ab}=0$. Thus besides Einstein spacetimes, for which the Bach tensor vanishes identically, conformal-to-Einstein spacetimes solve conformal gravity field equations as well.

The remaining field equations are solved iff
\be
a_1 = \frac{{a_2}^2 {\rn}^4-1}{3 a_3 {\rn}^4}. \label{eqa1conf}
\ee
The function
\be
f=a_0 + a_1 r + a_2 r^2 + a_3 r^3 \label{congravKundt}
\ee
is thus a solution to conformal gravity provided \eqref{eqa1conf} holds.  By applying the transformations \eqref{tr_t} and \eqref{tr_r}, one can set $a_3$ equal to $1$ and another coefficient, e.g., $a_2$ or $a_0$, equal to $0$. Therefore, taking \eqref{eqa1conf} into account, the remaining free parameters are  $\rn$ and at most one integration constant. Specifically, if $a_0$ is set to zero, either $a_1$ or $a_2$ remains as a free parameter (since they are related by \eqref{eqa1conf}). Alternatively, if $a_2$ is set to zero, $a_1$ is fully determined by \eqref{eqa1conf}, leaving $a_0$ as the remaining integration constant.

In this case, the Bach tensor vanishes while the Ricci scalar $R$ is non-constant
\be
R=\frac{2}{{\rn}^2} -2 a_2  - 6 a_3 r. 
\ee
The non-constancy of $R$ implies that this spacetime is not an Einstein space and also that this metric, with nonconstant \(R\), lies outside the \(R=\mathrm{const.}\) class studied in \cite{PraPraPodSva21}, which assumed $R=$const. and primarily focused on expanding static spherically symmetric spacetimes (some Kundt spacetimes were found there as special cases). However, it is locally conformal to Schwarzschild-(A)dS spacetime.

Since ${f''=2 (a_2+3 a_3 r)}$, the polynomial curvature invariants \eqref{eqInvars} blow up at ${r\to\pm\infty}$ and the affine-parameter integral \eqref{eq:affpar} converges at ${r\to\infty}$ for ${a_3<0}$ and at ${r\to-\infty}$ for ${a_3>0}$ for all causal geodesics except radial null ones ${\epsilon=L=0}$; otherwise the limit is not attainable. Hence, ${r\to-\infty}$ is a geodesic singularity {(i.e., a curvature singularity reachable by a causal geodesic in a finite affine parameter)} for ${a_3>0}$, while ${r\to\infty}$ is a geodesic singularity  for ${a_3<0}$.

Note that the cases ${f''=\text{const.}}$, which are studied in the following sections~\ref{sec_a2}--\ref{sec_a0}, are direct products of constant curvature two-spaces and are regular spacetimes.

\subsection{Case $\alpha \neq  3 \beta$ and $f=a_0 + a_1 r + a_2 r^2 $ with $a_2 \neq  0$ }\
\label{sec_a2}

Since in this case ${f''=\text{const.}}$, the resulting spacetime is a direct product of constant curvature two-spaces. Consequently, the curvature tensors are covariantly constant
\be
	C_{abcd;e} = 0 = R_{ab;c}
\ee
and the Ricci scalar $R$ is constant.

The trace equation reduces to

\be
2 \gamma \left(a_2 - \frac{1}{{\rn}^2} + 2 \Lambda\right)=0 \label{eqtracefquad}
\ee
and this case naturally splits into two subcases.

\subsubsection{Subcase $\gamma \neq 0$, Bachian-Nariai and Bachian-Bertotti-Robinson spacetime}

For $\gamma \neq 0$, Eq. \eqref{eqtracefquad}
gives
\be
a_2=\frac{1-2\Lambda \rn^2}{\rn^2},\ \ \ \ 
 \label{pol2_r0}
\ee
and the remaining field equations are reduced to
\be
(a_2+\Lambda ) [3 \gamma -8 \Lambda  (\alpha -3 \beta )] =0.\label{cond_class}
\ee
Vanishing of the first bracket leads to the Nariai spacetime \eqref{Nariai}, with $a_2=-\Lambda=-\frac{1}{\rn^2}$. Using \eqref{tr_t}, \eqref{tr_r}, one can set $a_1=0$ and  thus the only free parameters are $\Lambda$ and $a_0$.
The vanishing of the second bracket implies
\be
\Lambda = \frac{3 \gamma }{8 (\alpha -3 \beta )} \,.\label{pol2_Lambda}
\ee
Note that this is equivalent to $\Lambda=\frac{3}{8 k}$ in the notation of \cite{PraPraPodSva21}.
Eqs. \eqref{pol2_r0} and \eqref{pol2_Lambda} are sufficient for the vacuum field equations to hold.
This direct-product metric has $R_{(I)}=-2a_2$ with an arbitrary sign. Therefore, for $2\Lambda \rn^2>1$, it is Bachian-Nariai and for $2\Lambda\rn^2<1$, it is Bachian-Bertotti-Robinson (see Table \ref{tableNariaietal}).
Using \eqref{tr_t}, \eqref{tr_r}, one can set $a_1=0$ and therefore, the only free parameters are  $\rn$ and $a_0$.

While the Bach tensor for this spacetime is non-vanishing and depends non-trivially on $r$ and $\theta$, its associated scalar invariant remains constant
\be
B_{ab} B^{ab} = \frac{\gamma ^2 [8 {a_2} (\alpha -3 \beta )+3 \gamma ]^2}{256 (\alpha -3 \beta )^4}
\ee
as well as the Ricci scalar. This is the Bachian-Nariai 
or Bachian-Bertotti-Robinson solution of \cite{PraPraPodSva21} depending on sign$(a_2)$.

\subsubsection{Subcase $\gamma =0$, Nariai and Bertotti-Robinson spacetimes}
\label{sec_f2_g0}

In the $\gamma=0$ case,  the field equations reduce to 
\be
	\left({a_2}^2 {\rn}^4 -1\right)  (\alpha - 3 \beta)=0.
\ee
This section assumes $\alpha \neq  3 \beta$ and consequently we are left with two cases
\be
a_2 = \pm \frac{1}{{\rn}^2},
\ee
for both of which the Bach tensor vanishes by \eqref{Bachvanishescond}. These spacetimes are direct products of constant-curvature spaces with  $R_{(I)} = \mp \frac{2}{\rn^2}$, and $R_{(II)}=\frac{2}{\rn^2} $ (cf. \eqref{curv-2spaces}). Thus,
the case $a_2 = - \frac{1}{{\rn}^2}$ is Nariai and the case with 
$
a_2 =  \frac{1}{{\rn}^2} 
$
is Bertotti-Robinson (see Table \ref{tableNariaietal}). Similarly as above, in both cases, the only free parameter is $\rn$.

\subsection{{Case $\alpha \neq  3 \beta$ and $f=a_0 + a_1 r $ with $a_1 \neq  0$, Pleba\'nski-Hacyan spacetime  } }
\label{sec_a1}

In this case, the field equations of quadratic gravity hold iff
\be
\label{CaseF2quad}
{\rn}^2 = \frac{4 (\alpha -3 \beta )}{3 \gamma }, \ \ \ \Lambda = \frac{3 \gamma }{8 (\alpha -3 \beta )}.
\ee
Using \eqref{tr_t}, \eqref{tr_r}, one can set $a_1=1$ and $a_0=0$ and thus there are no free parameters.

This spacetime is again a direct product of constant-curvature spaces. In this case, the curvature $R_{(I)}$ vanishes \eqref{curv-2spaces}, and thus this is the Plebański--Hacyan spacetime (see Table \ref{tableNariaietal}). From \eqref{Bachvanishescond}, the Bach tensor is non-zero; thus, the Einstein tensor is necessarily non-vanishing as well. This spacetime was identified as a vacuum solution to quadratic gravity in \cite{PraPraPodSva21}.

\subsection{Case $\alpha \neq  3 \beta$ and $f$ constant, Pleba\'nski-Hacyan again }
\label{sec_a0}

For constant $f$, the field equations of quadratic gravity hold if and only if the conditions in \eqref{CaseF2quad} hold. This spacetime is a direct product of constant-curvature spaces with vanishing $R_{(I)}$ \eqref{curv-2spaces}, and thus it is again the Plebański--Hacyan metric.

\subsection{Theories with $\alpha= 3 \beta$}
\label{sec_a3b_QG}

The main difference from the previous case $\alpha \neq 3 \beta$ is that for theories with $\alpha = 3 \beta$, we cannot use Eq. \eqref{eqQGcomb} to obtain $f^{(4)} = 0$. Thus, in this case, $f$ is generally non-polynomial.

The $E_{tt}$ field equation \eqref{FEQEtt} reduces to
\be
\label{3.5fe00}
 \gamma  {\rn}^2 \left(\Lambda  {\rn}^2-1\right) = 0\,. 
\ee

Let us first study the $\gamma =0$ case.

\subsubsection{Case $\gamma =0$ and $\alpha= 3 \beta$}
\label{sec_a3bg0}

For $\gamma =0$, the remaining field equation \eqref{eqEthth}
reduces to
\be
f f^{(4)}+f^{(3)} f' =0. \label{caseA3BG0eq}
\ee
 Integrating the equation once with respect to $r$ yields

 \be
 f\,f^{(3)} = C. \label{firstinteg}
 \ee
For $C=0$, this implies $f^{(3)}=0$, and thus (see Sec. \ref{sec_f2_g0})
\be
f=a_0+a_1 r+a_2 r^2\,.
\ee
There are no further constraints arising from the field equations. Thus, for $C=0$, all direct product metrics given in Table \ref{tableNariaietal} are solutions. Using \eqref{tr_t}, \eqref{tr_r}, one can set as above $a_1=0$  and the solution has two free parameters  $a_2$ and $a_0$.

For the rest of this section, let us assume $C\neq 0$.
Equation \eqref{firstinteg} can be integrated once more to yield
\be
 f f'' - \frac{1}{2}\left(f'\right)^2  = Cr + D.
\ee
Now, using the substitution $f(r)=y(r)^2$ this equation reduces to
\be
2 y^3 y''= C r + D.
\ee
Finally, the substitution $x=Cr+D$ yields 
\be
y''=\frac{1}{2 C^2} x y^{-3}.
\ee
This is a special case of the Emden-Fowler equation
\be
y''=A x^n y^m,
\ee
with $n=1$, $m=-3$. These particular values of $m$ and $n$ do not belong to the known solvable cases of the Emden-Fowler equation  \cite{Polyaninbook}.

A particular solution is 
$$
y= \lambda (-x)^{3/4},
$$
which corresponds to the particular solution of \eqref{caseA3BG0eq}
\be
f(r) = a_0 (r - r_0)^{3/2} . \label{partsol-caseA3BG0eq}
\ee
Note that for this solution, the polynomial curvature invariants diverge at ${r=r_0}$ because ${f''={3 a_0}/{(4 \sqrt{r-r_0})}}$, implying a curvature singularity. From the convergence of \eqref{eq:affpar}, we see that it is reachable by a geodesic at a finite value of the affine parameter.

Using \eqref{tr_t}, \eqref{tr_r}, one can set $a_0=\pm 1$
and the solution has only one free parameter $\rn$.

In Sec. \ref{sec_alpha3b}, we briefly return to solving \eqref{caseA3BG0eq} using a generalized power series with arbitrary leading power and integer steps. This yields \eqref{partsol-caseA3BG0eq}, cf. \eqref{resf_32}, and other power series. The former solution 
suggests that more general expansions, such as Puiseux series with fractional steps or series involving logarithmic terms, would be needed to study the general solution  of \eqref{caseA3BG0eq}. This is beyond the scope of this paper.

\subsubsection{Case $\gamma \neq 0$ and $\alpha= 3 \beta$}
\label{sec_a3bgne0}

For $\gamma \neq 0$, Eq. \eqref{3.5fe00} yields 
\be
{\rn}^2 = \frac{1}{\Lambda}\label{eq_r0}
\ee
and the remaining field equation \eqref{eqEthth} reduces to
\be
2 \gamma  \Lambda +\gamma  f''-6 \beta  \left(f f^{(4)}+f^{(3)} f' \right) =0.\label{casealpha3betaeqf}
\ee
 Integrating the equation once with respect to $r$ yields

 \be
 \gamma\,f' + 2\gamma\Lambda\,r - 6\beta\,f\,f^{(3)} = C.
 \label{eq_a3b_1}
 \ee
Rewriting the non-linear term as a total derivative, we can integrate the equation once more to obtain the second integral
\be
\gamma\,f + \gamma\Lambda\,r^2 - 6\beta\left[ f f'' - \frac{1}{2}\left(f'\right)^2 \right] = Cr + D,
\label{eqquada3b}
\ee
where $D$ is a constant of integration.

An exact particular solution of \eqref{eqquada3b} is
\be
f = -\Lambda r^2 + \left( \frac{C}{\gamma} \right) r + \frac{D - \frac{3\beta C^2}{\gamma^2}}{\gamma + 12\beta\Lambda},
\ee
assuming $\gamma + 12\beta\Lambda\not=0$. Taking into account also \eqref{eq_r0},  this corresponds to the Nariai spacetime (with only one free parameter $\Lambda$). 

To find a more general solution, we will resort to a power series expansion.

\subsubsection{Case $\alpha= 3 \beta$, power series solutions}
\label{sec_alpha3b}

Due to the necessary condition $f^{(4)}=0$ in the $\alpha \neq  3 \beta$ case (see Eq. \eqref{eqQGcomb}), it was possible to obtain all solutions in a closed form within this class.

In contrast, for $\alpha= 3 \beta$, we have identified only particular solutions. To complement the study of the $\alpha= 3 \beta$ case, we employ a generalized power series around an arbitrary point $r_0$, adopting an arbitrary leading power $p$ followed by integer steps:
\be 
f=\Delta^p\sum_{i=0}^\infty a_i \Delta^{i},\ \  \Delta\equiv r-r_0,\ \  p\in \mathbb{R},\ \ a_0\neq 0.\label{series_f}
\ee
Note that by using \eqref{tr_r}, one can arbitrarily shift the expansion point $r_0$. Substituting \eqref{series_f} into  \eqref{casealpha3betaeqf}, we  obtain
\be
2\gamma\Lambda+\gamma
\sum_{l=p-2}^\infty\Delta^l a_{l-p+2}(l+2)(l+1)
-6\beta \sum_{l=2p-4}^\infty \sum_{j=0}^{l-2p +4}\Delta^l 
a_j a_{l-2p+4-j}(p+j)(p+j-1)(p+j-2)(l+1)=0.
\label{eq_series_f}
\ee
We first determine the values of $p$ that are compatible with \eqref{eq_series_f} by solving the resulting indicial equation.\\

$\bullet$ $p>2$ leads to $\Lambda=0$ and $a_0 p (p-1)=0$, and thus, $a_0=0$. Therefore the case $p>2$ is not allowed.\\

$\bullet$ $p=2$ $\rightarrow$ the order 0
of Eq. \eqref{eq_series_f}
 gives $2\gamma \Lambda +2 \gamma a_0 =0$ and thus $ a_0=-\Lambda$;\\

 $\bullet$ $p<2$ $\rightarrow$ the lowest order of Eq. \eqref{eq_series_f}
 gives $ \beta  a_0^2 p (p-1)(p-2)(2p-3)=0$ and thus there are the following possibilities:  $p=0,1, 3/2$.\\

  Now, let us examine the allowed cases in more detail. \\

$\bullet$ {\bf{case}} $p=2$\\
Eq. \eqref{eq_series_f}
 gives 
  \be
 a_0=-\Lambda.\footnote{This assumes $\gamma\not=0$. For $\gamma=0$, the quadratic coefficient is arbitrary.}
 \ee
For generic $\Lambda$, Eq. \eqref{eq_series_f} then implies $a_i=0$ for all
$i\geq 1$, which corresponds to the Nariai spacetime.

However, for special values of $\Lambda$
 \be
 \Lambda=-\frac{\gamma}{6\beta L (L+1)},  \ \ L \in \mathbb{N}, \label{Lambda_spec}
 \ee
$a_i=0$ for $1\leq i<L$, $a_L$ is arbitrary
 and for $i>L$, the coefficients $a_i$ are determined by the recurrence relation
 \be
 a_l=\frac{6\beta\sum_{j=1}^{l-1}a_j a_{l-j}(j+2)(j+1)j}{(l+2)[\gamma+6\beta\Lambda l (l+1)]}.
 \label{eq_cj}
 \ee
 Note that $a_l$ as given by \eqref{eq_cj} is nonzero for integer multiples of $L$.
 This class has one free parameter $a_L$.

This solution seems to be a Kundt counterpart of the expanding extreme higher-order (discrete) Schwarzschild-Bach-de Sitter solution found in \cite{PraPraPodSva21} with $\Lambda=-\frac{3\gamma}{[6L (L+1)-8]\alpha+24\beta}$,
which for $\alpha=3\beta$ reduces to \eqref{Lambda_spec}.

 $\bullet$ {\bf{case}} $p=3/2$\\
 In this case, we obtain $\gamma a_0=0$. Since $a_0=0$ is not allowed, it follows that $\gamma=0$ and the series solution truncates to one term 
 \be
 f=a_0 \Delta^{3/2}.\label{resf_32}
 \ee
 This is the particular solution \eqref{partsol-caseA3BG0eq}
 found in Sec. \ref{sec_a3bg0}.\\

 $\bullet$ {\bf{case}} $p=1$\\
 For $p=1$,  Eq. \eqref{eq_series_f}  gives coefficients for the lower-order terms:
 \bea
 a_2\rovno \frac{\gamma(a_1+\Lambda)}{18\beta a_0}\,,\\
 a_3\rovno -\frac{\gamma(a_1+\Lambda)(12\beta a_1-\gamma)}{864 \beta^2 a_0^2}\,,\\
 a_4\rovno \frac{\gamma(a_1+\Lambda)(432\beta^2a_1^2-72\beta\gamma a_1+\gamma^2-24\beta\gamma \Lambda)}{77760 \beta^3 a_0^3}\,,\dots
 \eea
The general coefficient $a_{l+2}$ is determined by the recurrence relation
\be
a_{l+2}=\frac{\gamma a_{l+1}(l+2)-6\beta\sum_{j=1}^{l+1}
a_j a_{l+2-j}j(j+1)(j-1)}{6\beta a_0 (l+3)(l+2)(l+1)},\ \ l \geq 1.
\ee

This class possesses two free parameters, $a_0$ and $a_1$. However, using \eqref{tr_t}, \eqref{tr_r}, one can set $a_0=1$ and $r_0=0$ and  reduce free parameters to $a_1$ and $\Lambda$.  It reduces to the Nariai spacetime for $a_1=-\Lambda$, in which case all $a_i$ for $i\geq 2$ vanish. For $\gamma=0$ (and $\beta \neq 0$), the series truncates after the quadratic term and the solution reduces to the spacetimes given in Tables \ref{tableNariaietal}/\ref{Tab_constR}: (Bachian-)Nariai, (Bachian-)Bertotti-Robinson, and Pleba\'nski-Hacyan. \\

$\bullet$ {\bf{case}} $p=0$ \\

For $p=0$, Eq. \eqref{eq_series_f} implies
\bea
a_4\rovno -\frac{a_1 a_3}{4a_0}+\frac{\gamma (a_2+\Lambda)}{72 \beta a_0}\,,\\
a_5\rovno -\frac{a_3(a_0 a_2-a_1^2)}{10 a_0^2}
+\frac{\gamma [3a_0a_3-2 a_1 (a_2+\Lambda)]}{360 \beta a_0^2}\,,\\
a_6\rovno -\frac{a_3^2}{20 a_0}
+\frac{a_1 a_3 [9\beta (2a_0a_2-a_1^2)-\gamma a_0]}{180 \beta a_0^3}
-\frac{\gamma(a_2+\Lambda)[36\beta (a_0a_2-a_1^2)-\gamma a_0]}{12960 \beta^2 a_0^3}, \ \dots\,
\eea
$a_5$ and all subsequent coefficients  are determined by the recurrence relation
\be
a_{l+4}=\frac{\gamma a_{l+2}(l+2)
-6\beta\sum_{j=1}^{l+3}a_ja_{l+4-j}j(j-1)(j-2)}{
6\beta a_0 (l+4)(l+3)(l+2)}, \ \ \ l \geq 1.
\ee
This solution possesses four parameters $a_0,a_1,a_2,a_3$.
Using \eqref{tr_t}, \eqref{tr_r}, one can set $r_0=0$ and, e.g., $a_0=\pm 1$ or $a_1=1$, reducing thus the number of free parameters to three integration constants plus $\Lambda$.
For $a_2=-\Lambda$ and $a_3=0$, the solution reduces to Nariai (all $a_i=0$ for $i\geq 3$). For $a_3=0$ and $\gamma=0$,  the series truncates after the quadratic term, and the solution again reduces to the spacetimes given in Tables \ref{tableNariaietal}/\ref{Tab_constR}.

 Note that for $a_0 \neq  0$, the analyticity conditions of Cauchy-Kowalewski for \eqref{eqquada3b} are fulfilled, and the formal solution should converge to an actual solution of  \eqref{eqquada3b} in some neighborhood of $r=0$.
This new solution admits a non-trivial Bach tensor and a  non-constant Ricci scalar $R$, cf. \eqref{Bachcond}, \eqref{eqricci}.

Let us conclude this section by summarizing all possible vacuum static spherically symmetric  Kundt metrics in quadratic gravity for various combinations of 
parameters of the theory in Tab. \ref{tab_allSol}.
It is interesting to note that such solutions exist only if $\Lambda\neq 0$ or $\gamma =0$. For this reason, no Kundt solution was found in~\cite{Podolskyetal20,Giacchini:2025mlv}, which only considered the case of $\Lambda=0$ and $\gamma\neq 0$.

\begin{tiny}
\begin{table}[h]
		\begin{center}
			\begin{tabular}{|l||l|l|c|c|c|}
				\hline
Theory & Metric function $f$ & Conditions & $R=$ & Sec. & Spacetime \\[1mm] \hline\hline
 $\alpha$, $\beta$ arb., $\gamma, \Lambda\neq 0$ &
 $a_0+a_1 r+a_2 r^2$
 &  $\rn^2=\frac{1}{\Lambda}=-\frac{1}{a_2}$  
 & c. &  \ref{sec_a2}& Nariai \\
$\alpha\neq 3\beta$, $\gamma\neq 0$, $\Lambda=\frac{3\gamma}{8(\alpha-3\beta)}$ &
 $a_0+a_1 r+a_2 r^2$
 & $\rn^2=\frac{1}{a_2+2\Lambda}$   & c. &\ref{sec_a2} & BN/BBR\\[1mm]
$\alpha\neq 3\beta$, $\gamma\neq 0$,
$\Lambda=\frac{3\gamma}{8(\alpha-3\beta)}$
&  $a_0+a_1 r$ 
&  $\rn^2=\frac{4(\alpha-3\beta)}{3\gamma}$  &
c. &
\ref{sec_a1} & PH \\[1mm]
$\alpha\neq 3\beta$, $\gamma\neq 0$,
$\Lambda=\frac{3\gamma}{8(\alpha-3\beta)}$
&  $a_0$ 
&  $\rn^2=\frac{4(\alpha-3\beta)}{3\gamma}$  &
c. &
\ref{sec_a0} & PH  \\[1mm] \hline
$\alpha$, $\beta$ arb.,  $\gamma=0$
& $a_0+a_1 r+a_2r^2$
& $a_2=\pm\frac{1}{\rn^2}$ & c. & \ref{sec_a2} & BR/Nariai
				\\[1mm] \hline
$\alpha$ arb., $\beta=0=\gamma$
&  $a_0+a_1 r+a_2 r^2+a_3r^3$
& $a_1=\frac{a_2^2 \rn^4-1}{3 a_3 \rn^4}$ & nc. &
                \ref{sec_a3} & \\[1mm] \hline
$\alpha=3\beta$ arb., $\gamma\not=0$, $\Lambda>0$ 
& $a_0 r^2$ 
& $\rn^2=1/\Lambda$, $a_0=-\Lambda$
& c. & \ref{sec_alpha3b} & Nariai \\
$\alpha=3\beta$ arb., $\gamma\not=0$,  $\Lambda =-\frac{\gamma}{6\beta L (L+1)}>0$
& \eqref{series_f}, $p=2$
 &$\rn^2=1/\Lambda$, $a_0=-\Lambda$ 
 & nc. & 
    \ref{sec_alpha3b}&  ho Bachian Kundt\\
$\alpha=3\beta$ arb., $\gamma\not=0$,  $\Lambda >0$
&
$a_0 r +a_1 r^2$
 &$\rn^2=1/\Lambda$, $a_1=-\Lambda$
 & c. &\ref{sec_alpha3b} & Nariai\\
$\alpha=3\beta$ arb., $\gamma\not=0$,  $\Lambda >0$
& \eqref{series_f}, $p=1$
 &$\rn^2=1/\Lambda$, 
 $a_2=\frac{\gamma (a_1+\Lambda)}{18\beta a_0}$& nc. & \ref{sec_alpha3b}& Bachian Kundt\\
$\alpha=3\beta$ arb., $\gamma\not=0$,  $\Lambda >0$
& $a_0+a_1 r+a_2r^2$
 &$\rn^2=1/\Lambda$,
 $a_2=-\Lambda$
 & c. &\ref{sec_alpha3b} & Nariai\\
$\alpha=3\beta$ arb., $\gamma\not=0$,  $\Lambda >0$
& \eqref{series_f}, $p=0$
 &$\rn^2=1/\Lambda$
& nc. & \ref{sec_alpha3b}& Bachian Kundt \\
  $\alpha=3\beta$, $\gamma=0$ &
 $a_0+a_1r+a_2 r^2$& & c. & 
 \ref{sec_a3bg0} 
 & (B)N/(B)BR/PH \\ 
   $\alpha=3\beta$, $\gamma=0$ &
$a_0 r^{3/2}$  & & nc. & \ref{sec_a3bg0} &\\
 $\alpha=3\beta$, $\gamma=0$ 
 & \eqref{series_f}, $p=0$
 & $a_0\not=0\not=a_3$ 
 & nc. & \ref{sec_alpha3b}& Bachian Kundt \\
\hline
			\end{tabular} \\[2mm]
			\caption{ Summary of all static spherically symmetric  Kundt metrics as vacuum solutions to quadratic gravity  (c/nc stands for constant and nonconstant,
            ho stands for higher-order, arb. stands for arbitrary, (B)N stands for (Bachian-)Nariai, PH stands for Pleba\'nski-Hacyan, and  BR and (B)BR stand for  Bertotti-Robinson and  (Bachian-)Bertotti-Robinson, respectively).  }
            \label{tab_allSol}
\end{center}
	\end{table}
\end{tiny}



In Table \ref{Tab_constR}, we summarize $R=$const. static, spherically symmetric solutions in quadratic gravity (corresponding to $f(r)$ being at most quadratic in $r$). Here, more technical details are given than in Table \ref{tableNariaietal}.  {In fact, all solutions in Table \ref{Tab_constR}  were already found in the conformal-to-Kundt coordinates in \cite{PraPraPodSva21} in  the classes $[0,0]^\infty$, $[0,1]^\infty$, $[0,2]^\infty$ therein. In addition, the Nariai and Bachian-Nariai/Bachian-Bertotti-Robinson spacetimes also correspond to the Kundt limits of expanding classes $[0,0]$, $[0,1]$, $[0,2]$ of \cite{PraPraPodSva21}.}

\begin{tiny}
\begin{table}[h]
		\begin{center}
			\begin{tabular}{|c|c|l|c|c|l|c|}
				\hline
$R_{I}$ & geometry & $a_2$ & $R$  & $  S_{01}$ & Conditions on parameters& Spacetime\\[1mm] \hline\hline
$+$ & $dS_2\times S^2$ & $<0$ & $+$  & $\neq 0$ & $\Lambda=\frac{3\gamma}{8(\alpha-3\beta)}>0$
or $\gamma=0=\alpha-3\beta$& Bachian-Nariai\\
$+$ & $dS_2\times S^2$ & $=-\frac{1}{\rn^2}<0$ & $+$ & $0$ & 
$\Lambda=\frac{1}{\rn^2}>0$ or $\gamma=0$ & Nariai\\[3mm] \hline
$0$ & $M_2\times S^2$ & $0$ & $+$ & $+$ &  ($\rn^2=\frac{1}{2\Lambda}$, $\Lambda=\frac{3\gamma}{8(\alpha-3\beta)}$)
or $\gamma=0=\alpha-3\beta$& Pleba\'nski-Hacyan\\[3mm] \hline 
$-$ & $AdS_2\times S^2$ & 
$\left(0,\frac{1}{\rn^2}\right)$
& $+$  & $+$ &$\Lambda=\frac{3\gamma}{8(\alpha-3\beta)}$
or $\gamma=0=\alpha-3\beta$ & Bachian-Bertotti-Robinson\\
$-$ & $AdS_2\times S^2$ & $=\frac{1}{\rn^2}>0$ & $0$ & $+$ &  $\gamma=0$  & Bertotti-Robinson\\
$-$ & $AdS_2\times S^2$ & 
$>\frac{1}{\rn^2}$ & $-$ & $+$ &  $\Lambda=\frac{3\gamma}{8(\alpha-3\beta)}$
or $\gamma=0=\alpha-3\beta$ & Bachian-Bertotti-Robinson\\ \hline
\end{tabular}\\[2mm]
\end{center}
			\caption{Summary  of $R=$const.  static spherically symmetric  Kundt metrics as vacuum solutions to quadratic gravity.            
            These are direct product spacetimes, the Ricci scalar and tensor are decomposable, i.e., $R=R_I+R_{II}$, where ${R_I=-f''=-2a_2}$
            and ${R_{II}=\frac{2}{\rn^2}}$, so that, ${R=\frac{2}{\rn^2}-f''=2 \left(\frac{1}{\rn^2}-a_2\right)}$ and it vanishes only for $a_2=\frac{1}{\rn^2}$, the Bertotti-Robinson spacetime. Then also ${C_{abcd}C^{abcd}=0}$ and the spacetime is conformally flat. The traceless Ricci component is ${S_{01}=\frac{1}{4}\left(\frac{2}{\rn^2}+f''\right)=\frac{1}{2}\left(\frac{1}{\rn^2}+a_2\right)}$,
            cf. \eqref{eqS01}, and vanishes only for ${a_2=-\frac{1}{\rn^2}}$ for the Nariai spacetime. In the cases when $\gamma\neq0$, then $R=4\Lambda$.}
            \label{Tab_constR}
\end{table}
\end{tiny}


\section{Six-derivative gravity}\label{sc:SDG}

\label{sec_a3b}

In this section, we illustrate the possible diversity of static spherically symmetric Kundt solutions to gravity models beyond quadratic gravity by focusing on some cases of six-derivative gravity.  In particular, we highlight how the conditions $\Lambda \neq 0$ or $\gamma =0$ can be relaxed and still produce the same geometry.

\subsection{Field equations}
In four dimensions, a generic six-derivative gravity model can be described by the action
\bea
S & = & \int \dd^4 x\, \sqrt{-g}\, \Big[
\gamma \,(R-2\Lambda) +\beta\,R^2  - \alpha\, C_{abcd}\, C^{abcd}	
- \eta_1 \, C_{abcd} \, \Box \, C^{abcd} 
+ \eta_2 \, R \, \Box \, R 
\nonumber
\\
&&
+ \eta_3 \, R^3 
+ \eta_4 \, R \, S_{ab} \, S^{ab} 
+ \eta_5 \, S_{ab} \, S^a{}_c \, S^{cb}
+ \eta_6 \, S_{ab} \, S_{cd} \, C^{acbd} 
+ \eta_7 \, R \, C_{abcd} \, C^{abcd}
\nonumber
\\
&&
+ \, \eta_8 \, C_{abcd} \, C^{abef} \, C^{cd}{}_{ef} \Big]\,,
\label{action-6der}
\eea
where we supplemented the fourth-order gravitational action by terms containing six derivatives of the metric, with coupling constants $\eta_{1,\ldots,8}$. Here, $S_{ab} = R_{ab} - \frac{1}{4} g_{ab} R$ is the traceless Ricci tensor. Any other sixth-order term can be recast as a combination of the terms in~\eqref{action-6der} and boundary or topological terms~\cite{Decanini:2007gj}, that do not contribute to the field equations.

Using the ansatz~\eqref{SSKundtmetric} for the metric, the field equations of generic six-derivative gravity read
\begin{small}
\bea
E_{tt}: \quad
&&
36  \rn^4  \gamma \left(1-\Lambda  \rn^2\right) 
-24 \rn^2 (\alpha -3 \beta ) 
+12 \left(12 \eta _3+3 \eta _4+\eta _6+4 \eta _7\right)+8 \eta _8
\nonumber
\\
&& 
+3 \rn^4 \left[ 2 \rn^2 (\alpha -3 \beta ) -36 \eta _3+3 \eta _4+\eta _6-12 \eta _7 -2\eta _8  \right]  \big( f''^2  - 2 f' f^{(3)} - 4 f  f^{(4)} \big)
\nonumber
\\
&&
+ \rn^6 \left(36 \eta _3+9 \eta _4+3 \eta _6+12 \eta _7+2 \eta _8\right) \left[ f''^3 - 3 \big(  f' f^{(3)} + 2 f f^{(4)} \big) f'' - 6 f f^{(3)2}\right]  
\nonumber
\\
&&
- 6  \rn^6 \left(\eta_1 - 3 \eta_2 \right) \left[ 4  f'^2 f^{(4)} + \big( 3 f^{(3)2} + 10 f''  f^{(4)} + 14 f' f^{(5)} + 4 f f^{(6)} \big) f \right] 
= 0,
\label{6der-eom-tt}
\eea
\bea
E_{rr}:
\quad
&& 
36  \rn^4 \gamma \left(1-\Lambda  \rn^2\right) 
-24 \rn^2 (\alpha -3 \beta ) 
+12 \left(12 \eta _3+3 \eta _4+\eta _6+4 \eta _7\right)+8 \eta _8
\nonumber
\\
&& 
+3 \rn^4 \left[ 2 \rn^2 (\alpha -3 \beta ) -36 \eta _3+3 \eta _4+\eta _6-12 \eta _7 -2\eta _8  \right]  \big( f''^2  - 2 f' f^{(3)} \big)
\nonumber
\\
&&
+ \rn^6 \left(36 \eta _3+9 \eta _4+3 \eta _6+12 \eta _7+2 \eta _8\right) \big( f''^3 - 3   f'  f'' f^{(3)} \big)  
\nonumber
\\
&&
- 6  \rn^6 \left(\eta_1 - 3 \eta_2 \right) \left[  4  f'^2 f^{(4)} + \big( f^{(3)2} - 2 f''  f^{(4)}  + 2 f' f^{(5)}  \big) f \right] 
= 0,
\label{6der-eom-rr}
\eea
\bea
E_{\theta \theta}: 
\quad
&& 
-72 \rn^6 \gamma  \Lambda  + 48 \rn^2 (\alpha - 3 \beta ) -16 \left(36 \eta _3+9 \eta _4+3 \eta _6+12 \eta _7+2 \eta _8\right)
\nonumber
\\
&&
-12 \rn^2 \left[ 3 \rn^4 \gamma  -36 \eta _3+3 \eta _4+\eta _6-12 \eta _7 -2\eta _8  \right] f''
-12 \rn^6 (\alpha -3 \beta ) f''^2
\nonumber
\\
&&
+12 \rn^4 \left[ 2 \rn^2 (\alpha +6 \beta ) -4 \eta _1+12 \eta _2+72 \eta _3-6 \eta _4+\eta _6-2 \eta _8 \right]  \big( f' f^{(3)} + f f^{(4)} \big)
\nonumber
\\
&&
-\rn^6 \left(36 \eta _3+9 \eta _4+3 \eta _6+12 \eta _7+2 \eta _8\right) f''^3
+36  \rn^6 \left(\eta _1+3 \eta _2\right) f f^{(3)2}
\nonumber
\\
&&
-6  \rn^6 \left(72 \eta _3+6 \eta _4-\eta _6-2 \eta _8\right) \left[  f f^{(3)2} + \big( f' f^{(3)}  + f f^{(4)}  \big) f'' \right] 
\nonumber
\\
&&
+24 \rn^6 \left(\eta _1+6 \eta _2\right) \left[ 2 f'^2 f^{(4)} + \big(  f'' f^{(3)}  + 4 f  f^{(5)} \big) f' + \big( 3  f''  f^{(4)}  + f f^{(6)} \big) f \right] = 0.
\label{6der-eom-thth}
\eea
\end{small}
Note, however, that the identity~\eqref{eq:Biattrr} shows that the components $tt$ and $rr$ are not independent and there are only two independent equations to be solved, $E_{rr}=0$ and $E_{\theta\theta}=0$.

In addition, for the combination $E_{tt}- E_{rr}$, we obtain
\begin{small}
\bea
&&
\big\lbrace
\rn^2 \left(2 \eta _1-6 \eta _2+36 \eta _3+9 \eta _4+3 \eta _6+12 \eta _7+2 \eta _8\right) f^{(3)2}
\nonumber
\\
&&
+ 2 \left[2 \rn^2 (\alpha -3 \beta )-36 \eta _3+3 \eta _4+\eta _6-12 \eta _7 -2\eta _8 \right]  f^{(4)} 
\nonumber
\\
&&
+ \rn^2 \left[ 12 \left(\eta _1-3 \eta _2\right) + 36 \eta _3+9 \eta _4+3 \eta _6+12 \eta _7+2 \eta _8 \right] f'' f^{(4)}
\nonumber
\\
&&
+4 \rn^2 \left(\eta _1-3 \eta _2\right) \big( f  f^{(6)} +3 f' f^{(5)} \big)
\big\rbrace f = 0 .
\label{eqSDGcomb}
\eea
\end{small}
While in quadratic gravity, such a combination resulted in a simple equation for the function $f(r)$ [see Eq.~\eqref{eqQGcomb}], in generic sixth-order gravity, it is considerably more complicated.

The structure of the field equations here and the eight extra parameters $\eta_{1,\ldots,8}$ make it difficult to carry out a thorough analysis of all the possible branches of solutions and models like the one performed in the previous section. Instead, our goal in this section is to explore some solutions to particular models in order to illustrate their variety and the differences to  quadratic gravity.

As a first example, notice that the field equations~\eqref{6der-eom-tt}--\eqref{6der-eom-thth} do not depend on the parameter $\eta_5$. This happens because, for a metric of the form~\eqref{SSKundtmetric}, the field equations resulting from the action
\be
\label{6der-eta5}
S = \int \dd^4 x\, \sqrt{-g}\, \eta_5 \, S_{ab} \, S^a{}_c \, S^{cb} 
\ee
are proportional to the derivatives of the constant $\rn$. As a consequence, any static spherically symmetric Kundt metric is a solution to the cubic model~\eqref{6der-eta5}.
A similar situation occurs for the particular models with
\be
\label{modtudoSDG}
\gamma = \eta_1 = \eta_2 = 0 , \quad
\alpha = 3 \beta \neq 0 , \quad
\eta_6 = -3 \eta_4 \neq 0, \quad
\eta_7 = - \frac{3}{4} ( 12 \eta_3 + \eta_4 ) , \quad
\eta_8 = \frac{9}{2} ( 8 \eta_3 + \eta_4 ) , \quad
\ee
for which the field equations~\eqref{6der-eom-tt}--\eqref{6der-eom-thth} are automatically satisfied for any $f(r)$ provided that $\rn^2=\eta_4/\beta$.

In the following Secs.~\ref{sec_lin}--\ref{SSec.Case3.6D}, we show how the geometries obtained in closed form in quadratic gravity can occur in six-derivative gravity, and afterwards, in~\ref{SSec.Case4.6D} and~\ref{SSec.Case5.6D}, we give some examples of solutions that only occur in the latter models. Accordingly, we shall only comment on the regularity of the solutions in the case of the new geometries, as the others were already considered in the previous section. We close this discussion with the identification of the possible classes of solutions expandable in power series according to the Frobenius method, in~\ref{sec_six_power}, followed by a brief note regarding (A)dS spacetime in six-derivative gravity.

\subsection{Case $f=a_0 + a_1 r $ with $a_1 \neq  0$}
\label{sec_lin}

In this case, compared to quadratic gravity, the inclusion of the six-derivative terms modifies Eq.~\eqref{CaseF2quad} to the form
\be
\label{CaseF2six}
\rn^2 =  \frac{2 ( \alpha - 3 \beta ) \pm \sqrt{4 (\alpha -3 \beta )^2-3 \gamma  \left( 36 \eta _3+9 \eta _4+3 \eta _6+12 \eta _7+2 \eta _8 \right)}}{3 \gamma } ,
\quad
\Lambda = -\frac{2 \left(\alpha -3 \beta -3 \gamma  \rn^2\right)}{9 \gamma  \rn^4} ,
\ee
Notice that such a solution might exist even if $\Lambda =0$, provided that
\be
\gamma (36 \eta _3+9 \eta _4+3 \eta _6+12 \eta _7+2 \eta _8)  = (\alpha -3 \beta )^2 .
\ee

\subsection{Case $f=a_0 + a_1 r + a_2 r^2 $ with $a_2 \neq  0$}
\label{SSec6der-2}

In this case, the field equations yield
\bea
\left(a_2 \rn^2+1\right) \left\lbrace 3 c_1 \left(a_2 \rn^2-1\right)^2 + c_2 \left[ 2 a_2 \rn^2-3 \left(a_2^2 \rn^4+1\right) \right] + 12 \rn^2  (\alpha -3 \beta ) \left(a_2 \rn^2-1\right) + 9 \gamma  \rn^4  \right\rbrace 
& = & 0,
\nonumber
\\
c_1 \left(a_2 \rn^2-1\right)^3-c_2 \left(a_2 \rn^2-1\right) \left(a_2 \rn^2+1\right)^2-9 \gamma  \rn^4 \left(a_2 \rn^2 - 1 + 2 \rn^2 \Lambda \right)
& = & 0,
\label{6der-2-sys}
\eea
where we defined
\be
\label{c1-c2}
c_1 = 2 \left( 18 \eta_3 + 6 \eta_7 + \eta_8 \right)  , \qquad c_2 = -3 \left( 3 \eta_4 + \eta_6 \right) .
\ee
This system undergoes considerable simplification if it is assumed that $a_2 \rn^2 = \pm 1$, i.e., the metric is Bertotti-Robinson or Nariai.

For instance, for the Nariai choice, $a_2 \rn^2 = -1$, Eq.~\eqref{6der-2-sys} reduces to
\be
\label{aquelaNar}
9 \gamma  \rn^4 \left(\Lambda  \rn^2-1\right) + 4 c_1 = 0 ,
\ee
which might admit real roots depending on the values of the parameters $\Lambda$, $\gamma$, and $c_1$. In particular, in the case $\Lambda = 0$, we have the solution
\be
\rn^2 =  \frac{2}{3} \sqrt{\frac{c_1}{\gamma}} , \qquad a_2 = - \frac{1}{\rn^2},
\ee
provided that $c_1 \gamma^{-1} > 0$.

On the other hand, for the Bertotti-Robinson choice, $a_2 \rn^2 = 1$, we obtain
\be
9 \gamma  \rn^4-4 c_2 = 0 , \qquad \gamma  \Lambda  \rn^6 = 0,
\ee
which can only be solved for a positive $\rn^2$ if $\Lambda = 0$ and $c_2 \gamma^{-1} > 0$, or if $\gamma=0=c_2$. In the former case, we obtain
\be
\label{B-R-SDG}
\rn^2 = \frac{2}{3} \sqrt{\frac{c_2}{\gamma}} , \qquad a_2 = \frac{1}{\rn^2}.
\ee

Therefore, in contrast to quadratic gravity, there is a large class of six-derivative gravity models that admit the Nariai spacetime as a solution even in the absence of a cosmological constant, and models for which the Bertotti--Robinson spacetime is a solution even in the absence of matter --- and for $\gamma \neq 0$. Moreover, in the latter case, such solutions do not depend on the couplings $\alpha$ and $\beta$ of the quadratic terms. In~\cite{Giacchini:2025gzw} it was shown that these solutions are the near-extreme-horizon limit of more general solutions in the so-called class $\lbrace 0,2 \rbrace$, which correspond to expansions around an extreme (double-degenerate) horizon.

There are models that admit the Bachian versions of these solutions, with $a_2 \rn^2 \neq \pm 1$. For instance, in the case $\Lambda = 0$ and $\alpha = 3 \beta$, we have the solution
\be
\rn^2 = \frac{\sqrt{2}}{3} \sqrt{\frac{c_1 c_2 \pm \sqrt{c_2^3 (2 c_1 - c_2) }}{(c_2 - c_1) \gamma }} , \qquad a_2 = -\frac{9 \gamma \rn^2}{2 c_2} ,
\ee
assuming that the parameters $c_1$ and $c_2$ satisfy the necessary constraints to guarantee that $\rn^2 > 0$. These solutions can be extended for the cases with generic $\alpha$ and $\beta$, as shown in~\cite{Giacchini:2025gzw}, and might also exist for non-vanishing $\Lambda$, depending on the values of the parameters of the model.

\subsection{Case $f=a_0 + a_1 r + a_2 r^2 + a_3 r^3$ with $a_3 \neq  0$ }
\label{SSec.Case3.6D}

In this case, Eq.~\eqref{eqSDGcomb} can only be satisfied if
\be
\label{condF3SDGb}
2 \eta _1-6 \eta _2+36 \eta _3+9 \eta _4+3 \eta _6+12 \eta _7+2 \eta _8  = 0 .
\ee
Although this condition considerably simplifies the field equations, their expansion order by order still yields a system of five equations that might have many branches of solutions depending on the possible relations between the parameters of the model. In what follows, we shall only discuss some particular solutions related to the condition $\eta_1 = 3 \eta_2$. In this context, one can verify that there exists the solution 
\be
\label{Sol.3.6D}
\rn^2 = \frac{2}{3} \sqrt{\frac{c_2}{\gamma}} , \qquad
a_1 = \frac{{a_2}^2 {\rn}^4-1}{3 a_3 {\rn}^4} ,
\ee
with $c_2$ defined in~\eqref{c1-c2},
provided that
\be
\Lambda = 0 , 
\qquad \gamma \neq 0 , 
\qquad \beta = 0 , 
\qquad \eta_1 = 3\eta _2=\frac{1}{6} \left(6 \eta _4-\eta _6\right), 
\qquad \eta _3=- \frac{1}{108} \left(3 \eta _4+4 \eta _6+12 \eta _7\right) ,
\label{condF3SDG}
\ee
in addition to~\eqref{condF3SDGb}.
Notice that the models defined by these conditions are contained in the class of models that admit the Bertotti-Robinson solution~\eqref{B-R-SDG}, with the same $\rn$. Also, while in quadratic gravity, a solution of this type only exists in Weyl conformal gravity [cf.~\eqref{eqa1conf}], here it can occur if the Einstein-Hilbert term is present in the action.

Solutions of this type to models with cosmological constant also exist. For instance, if we keep all the conditions~\eqref{condF3SDG} but allow $\Lambda \neq 0$ and set 
$\alpha = \tfrac{1}{2} \sqrt{\gamma c_2}$, we find the solution
\be
\label{Sol.3.6Db}
\rn^2 = \frac{2}{3} \sqrt{\frac{c_2}{\gamma}} , \qquad
a_1 = \frac{a_2^2 \rn^4+2 \Lambda  \rn^2-1}{3 a_3 \rn^4} .
\ee
To close, we mention that there are six-derivative models with $\gamma = 0$ that also admit this type of solutions, but carrying out an exhaustive analysis of all such possible models is beyond the scope of this work.
We recall that the analysis of the regularity is the same as in Sec.~\ref{sec_a3}, i.e., ${r\to-\infty}$ is a geodesic singularity for the case ${a_3>0}$, and so is ${r\to\infty}$ for ${a_3<0}$.

\subsection{Case $f = a_2 r^2 + a_n r^n$ with $a_n \neq  0$ }
\label{SSec.Case4.6D}

First, let us focus on a specific family of models defined by the following relations between the coupling constants:
\be 
\label{defmodelExSDG}
\alpha = 3 \beta , \qquad
\eta_1 = 3 \eta_2 , \qquad
\eta_6 = - 3 \eta_4 , \qquad
\eta _8 = -6 (3 \eta_3 + \eta_7).
\ee
These relations might be regarded as the equivalent of the relation $\alpha = 3 \beta$ in quadratic gravity, since they make Eq.~\eqref{eqSDGcomb} automatically satisfied independently of the function $f(r)$ [cf.~Eq.~\eqref{eqQGcomb}]. Moreover, in this case, the $tt$ component of the field equations reduces to
\be
\gamma  \rn^4 (\Lambda  \rn^2 -1) = 0,
\ee
which is the same as~\eqref{3.5fe00}. Therefore, like in quadratic gravity, there are two possibilities, $\gamma = 0$ or $\rn^2 = 1/\Lambda$. 
The latter case leads to the Nariai solution. So, let us assume that $\gamma = 0$.

By assuming that $f(r)$ is a polynomial function and solving the $\theta\theta$ component of the field equations order by order, one can verify that the following solutions exist, for some particular models:
\be
\label{SDG-sol-g0}
\begin{aligned}
f(r) &= a_2 r^2 + a_6 r^6 ,
& \rn^2 &= -\frac{24 \eta_2 - 5 \eta_4}{5 \beta + 76 a_2 \eta_2} ,
& \text{if} \quad
\eta_7 &= \frac{36 \eta_2}{5} - 9 \eta_3 - \frac{3 \eta_4}{4} , \\
f(r) &= a_2 r^2 + a_5 r^5 ,
& \rn^2 &= -\frac{21 \eta_2 - 5 \eta_4}{5 \beta + 39 a_2 \eta_2} ,
& \text{if} \quad
\eta_7 &= \frac{63 \eta_2}{10} - 9 \eta_3 - \frac{3 \eta_4}{4} , \\
f(r) &= a_2 r^2 + a_4 r^4 ,
& \rn^2 &= -\frac{10 \eta_2 - 3 \eta_4}{3 \beta + 8 a_2 \eta_2} ,
& \text{if} \quad
\eta_7 &= 5 \eta_2 - 9 \eta_3 - \frac{3 \eta_4}{4} , \\
f(r) &= a_2 r^2 + a_3 r^3 ,
& \rn^2 &= \frac{\eta_4 - 2 \eta_2}{\beta} ,
& \text{if} \quad
\eta_7 &= 3 \eta_2 - 9 \eta_3 - \frac{3 \eta_4}{4} .
\end{aligned}
\ee
These are particular cases of the solution
\be
\label{SolEx4N}
f(r) = a_2 r^2 + a_p r^p , \quad \rn^2 = - \frac{4 (p-2) (2 p-3) \eta _2 - (p-1) p \eta _4 }{(p-1) p \beta + (p-3) (p-2) \left(p^2+p-4\right) a_2 \eta _2 } , \quad p\in\mathbb{R}\setminus \{0,1\} , 
\ee
that exists for the model with
\be \label{Exa2+an}
\gamma = 0 , \quad
\alpha = 3 \beta , \quad
\eta_1 = 3 \eta_2 , \quad
\eta_6 = - 3 \eta_4 , \quad
\eta _7 = \frac{6 (p-2) (2 p-3) \eta _2 }{(p-1) p}  -9 \eta _3-\frac{3 \eta _4}{4}, \quad
\eta _8 = -6 (3 \eta_3 + \eta_7)  .
\ee
Notice that if $\eta_2=0$ then all solutions of the form $f(r) = a_2 r^2 + a_p r^p$ would exist for the same model; indeed, in this case we meet the family of models described by Eq.~\eqref{modtudoSDG}, for which any $f(r)$ is a solution.

The solutions described above also exist for $a_2=0$, so that the leading exponent becomes~$p$. In this case, it is possible to extend the family of models that admit this type of solution by relaxing the condition $\gamma=0$. Indeed, given any $p\in \mathbb{R}\setminus\lbrace0,1,\frac{3}{2},2,\frac{1}{10} (21\pm\sqrt{41})\rbrace$, the model with
\begin{subequations} \label{Exn=n}
\begin{eqnarray}
    \alpha &=& 9 \eta_4 \Lambda -\frac{\gamma  (p-1) (3 p-4) p^2}{36 \Lambda  (p-2) (2 p-3) [p (5 p-21)+20]}-\frac{12 \eta_1 \Lambda  (p-2) (2 p-3)}{(p-1) p} , 
    \\
    \beta &=& -\frac{\gamma }{6 \Lambda }+3 \eta_4 \Lambda -\frac{\gamma  (p-1) (3 p-4) p^2}{108 \Lambda  (p-2) (2 p-3) [p (5 p-21)+20]}-\frac{4 \eta_1 \Lambda  (p-2) (2 p-3)}{(p-1) p} ,
    \\
    \eta_2 &=& \frac{\eta_1}{3}+\frac{\gamma  (p-1) p (2 p-5)}{54 \Lambda ^2 (p-2) [p (5 p-21)+20]},  
    \\
    \eta_6 &=& -3\eta_4, 
    \\
    \eta_7 &=& -9 \eta_3-\frac{3\eta_4}{4}+\frac{\gamma  (7 p-15) [13 (p-4) p+48]}{216 \Lambda ^2 (p-2) [p (5 p-21)+20]}+\frac{2 \eta_1 (p-2) (2 p-3)}{(p-1) p},
        \\
    \eta_8 &=& 36 \eta_3+\frac{9 \eta_4}{2}+\gamma\frac{  p [(373-61 p) p-744]+480}{36 \Lambda ^2 (p-2) [p (5 p-21)+20]}-\frac{12 \eta_1 (p-2) (2 p-3)}{(p-1) p},
\end{eqnarray}
\end{subequations}
admits the solution
\begin{equation}
\label{Sol-Exn=n}
  f(r) = a_p r^p , \quad \rn^2 = \frac{1}{3\Lambda }.  
\end{equation}
If $\gamma = 0$, the model~\eqref{Exn=n} reduces to~\eqref{Exa2+an}, and the solution~\eqref{Sol-Exn=n} coincides with~\eqref{SolEx4N} with $a_2=0$. 

To close this example, let us mention that there are models that do not satisfy the conditions~\eqref{defmodelExSDG} and yet admit solutions of the form $f = a_2 r^2 + a_p r^p$ with $a_2\neq 0$ and $p>3$. 
Without assuming any of the special relations that reduce the order of field equations [such as~\eqref{modtudoSDG}, \eqref{defmodelExSDG} or simply $\eta_1=3\eta_2$], we find, for example, that the model with
\begin{equation}
\begin{gathered}
    \eta_1= -69 \eta_2 = \frac{391 (\alpha-3 \beta )^2}{2214 \gamma },  
    \quad\eta_4= \frac{2 (104 \alpha-5847 \beta ) (\alpha-3 \beta )}{16605 \gamma }, 
    \\
    \quad\eta_6= \frac{2 (\alpha-3 \beta ) (1201 \alpha+1932 \beta )}{5535 \gamma }, 
    \quad \eta_7= -\frac{(1987 \alpha-22566 \beta ) (\alpha-3 \beta )}{33210 \gamma }-9 \eta_3,
    \\
    \quad \eta_8= -\frac{(\alpha-3 \beta ) (4648 \alpha+2661 \beta )}{5535 \gamma } +36 \eta_3, 
    \quad\Lambda = -\frac{8883 \gamma }{1312 (\alpha-3 \beta )}
\end{gathered}
\end{equation}
admits the solution
\be
f(r) = \frac{15 \gamma  r^2}{8 (\alpha-3 \beta )} + a_4 r^4 , \quad \rn^2 = -\frac{2 (\alpha -3 \beta )}{3 \gamma }.
\ee

Regarding the regularity of the solutions discussed in this section, one can see that, for ${p>2}$, the polynomial curvature invariants diverge at ${r\to\pm\infty}$ (if it is defined there), since ${f''=2a_2 + a_p (p-1) p r^{p-2}}$. The affine-parameter integral \eqref{eq:affpar} converges, i.e., the singularity can be reached by geodesics at finite affine parameter, for real ${p>2}$ at ${r=\infty}$ if ${a_p<0}$ and at ${r=-\infty}$ (if defined there) if ${p}$ is an even integer for ${a_p<0}$ or $p$ is an odd integer for ${a_p>0}$. On the other hand, if ${p<2}$, then the polynomial curvature invariants diverge at ${r=0}$, in which case there is always a geodesic singularity.

\subsection{Case $f = a_0+a_1 r+ a_2 r^2 + c (r+b)^{5/2}$ with $c \neq  0$ }
\label{SSec.Case5.6D}

For models such that $\eta_1=3\eta_2$, Eq.~\eqref{eqSDGcomb} can be regarded as second-order differential equation for the function $f''(r)$, whose solution can be found in closed form, resulting in
\be
f(r) = a_0 + a_1 r - \frac{1}{30 x_1} \left[ 15 x_2 r^2 \pm \frac{2 \left[x_2^2 + 2 x_1 a_2 (r+a_3)\right]^{5/2}}{x_1^2 a_2^2} \right] ,
\ee
where
\begin{eqnarray}
x_1 & = & \left(36 \eta _3+9 \eta _4+3 \eta _6+12 \eta _7+2 \eta _8\right) \rn^2,
\\
x_2 & = & 4 \rn^2 (\alpha -3 \beta ) -72 \eta _3+6 \eta _4+2 \eta _6-4 \left(6 \eta _7+\eta _8\right).
\end{eqnarray}
It can then be verified that the $tt$-component of the field equations imposes a relation between the otherwise independent constants $a_{1,2,3}$, while further constraints will be derived from the $\theta\theta$ equation. Therefore, by taking the ansatz
\be
f(r) = a_0+a_1 r+ a_2 r^2 + c (r+b)^{5/2} , \label{eq:ansatz52}
\ee
we identified the solution 
\begin{equation}
    f(r)= -\frac{96 \gamma  (r+b)^2}{19 (\alpha -3 \beta )} + c (r+b)^{5/2}, \quad  \rn^2= \frac{19 (\alpha -3 \beta )}{96 \gamma },
\end{equation}
for the model with
\begin{equation}
\begin{gathered}
    \eta_1= 3 \eta_2,  \quad \eta_8= \frac{19 (\alpha -3 \beta ) (\alpha +141 \beta )}{3072 \gamma } -\frac{27 \eta_2}{8}+36 \eta_3, \quad \eta_7= \frac{19 (7 \alpha -165 \beta ) (\alpha -3 \beta )}{18432 \gamma } +\frac{9 \eta_2}{16}-9 \eta_3,
    \\
    \quad\eta_6= \frac{19 (31 \alpha -237 \beta ) (\alpha -3 \beta )}{4608 \gamma } -\frac{83 \eta_2}{20}, \quad\eta_4= -\frac{19 (7 \alpha -165 \beta ) (\alpha -3 \beta )}{13824 \gamma } + \frac{83 \eta_2}{60}, \quad\Lambda = -\frac{224 \gamma }{361 (\alpha -3 \beta )}.
\end{gathered}
\end{equation}

Also, for models with  $\eta_1=\eta_2=0$ we found
\begin{equation}
    f(r)=-\frac{8 \gamma }{\alpha -3 \beta } (r+b)^2+c (r+b)^{5/2},  \quad \rn^2=-\frac{4 (\alpha -3 \beta )}{13 \gamma },
\end{equation}
for the theory
\begin{equation}
\begin{gathered}
    \eta_1=\eta_2=0, \quad \eta _8=\frac{(\alpha -3 \beta ) (7 \alpha -93 \beta )}{52 \gamma }+36 \eta _3, \quad\eta _7=-\frac{(\alpha -3 \beta ) (19 \alpha -129 \beta )}{312 \gamma }-9 \eta _3, 
    \\
    \eta _6=\frac{(\alpha -3 \beta ) (\alpha +69 \beta )}{78 \gamma }, \quad \eta _4=\frac{(\alpha -3 \beta )(11 \alpha -105 \beta ) }{234 \gamma }, \quad  \Lambda =-\frac{247 \gamma }{24 (\alpha -3 \beta )}.
\end{gathered}
\end{equation}
Another similar one,
\begin{equation}
    f(r)=-\frac{16 \gamma  }{17 (\alpha -3 \beta )}(r+b)^2+c (r+b)^{5/2}, \quad \rn^2=\frac{17 (\alpha -3 \beta )}{16 \gamma } ,
\end{equation}
which solves
\begin{equation}
\begin{gathered}
    \eta_1=\eta_2=0, \quad \eta _8=-\frac{17 (\alpha -3 \beta ) (7 \alpha -93 \beta ) }{256 \gamma }+36 \eta _3, \quad\eta _7=\frac{17 (\alpha -3 \beta )(19 \alpha -129 \beta ) }{1536 \gamma }-9 \eta _3, 
    \\
    \eta _6=-\frac{17 (\alpha -3 \beta ) (\alpha +69 \beta )}{384 \gamma }, \quad\eta _4=-\frac{17 (\alpha -3 \beta )(11 \alpha -105 \beta ) }{1152 \gamma }, \quad \Lambda =\frac{304 \gamma }{867 (\alpha -3 \beta )}.
\end{gathered}
\end{equation}
Note that these solutions are for models with $\gamma \neq 0$ and, therefore, were not covered by the first example of Sec.~\ref{SSec.Case4.6D}.

Similar solutions also exist for models without a cosmological constant, for example,
\begin{equation}
    f(r)= \frac{3 \gamma  (r+b)^2}{2 (\alpha -3 \beta )} + c (r+b)^{5/2}, \quad  \rn^2= \frac{2 (\alpha -3 \beta )}{3 \gamma }
\end{equation}
solves
\begin{equation}
\begin{gathered}
    \eta_1= 3 \eta_2,  \quad \eta_8= -\frac{(37 \alpha -516 \beta ) (\alpha -3 \beta )}{135 \gamma }-\frac{249 \eta_2}{40}+36 \eta_3, \quad \eta_7= \frac{(361 \alpha -1893 \beta ) (\alpha -3 \beta )}{1620 \gamma }+\frac{83 \eta_2}{80}-9 \eta_3,
    \\
    \quad\eta_6= \frac{2 (2 \alpha -33 \beta ) (\alpha -3 \beta )}{27 \gamma }-\frac{9 \eta_2}{4}, \quad\eta_4= -\frac{(13 \alpha -93 \beta ) (\alpha -3 \beta )}{81 \gamma } + \frac{3 \eta_2}{4}, \quad\Lambda = 0.
\end{gathered}
\end{equation}

Considering the solutions in the form \eqref{eq:ansatz52}, we find ${f''=2 a_2+15 c \sqrt{r+b}/4}$, meaning that the polynomial scalar invariants blow up for ${r\to\infty}$. Using \eqref{eq:affpar}, this limiting value can be reached by causal geodesics for ${c<0}$ (unless ${\epsilon=L=0}$). Therefore, ${r\to\infty}$ is a geodesic singularity for ${c<0}$.

\subsection{Power series solutions }
\label{sec_six_power}

So far, we have presented several exact, closed-form solutions to six-derivative gravity in the form of a polynomial with either integer or fractional powers. In this section, we would like to present a different approach: finding solutions as a power-series expansion, as in Sec. \ref{sec_alpha3b}.
For this purpose, we expand each equation in powers of $\Delta\equiv r-r_0$ around an arbitrary point $r_0$ using
\be
f=\Delta^n\sum_{i=0}^\infty a_i\Delta^i,
\qquad a_0\neq 0
\label{six_power_f}
\ee
and compare the leading terms, which restrict the possible values of the leading power $n$.

Let us start with  the $E_{rr}$ field equation \eqref{6der-eom-rr}, which,  for clarity, can be written as
\be
\label{eq_rr44}
q_0
+q_1 \big( f''^2  - 2 f' f^{(3)} \big)
+ q_2 f''\big( f''^2 - 3   f'   f^{(3)} \big)  
+q_3 \left[  4  f'^2 f^{(4)} + \big( f^{(3)2}  -2 f''  f^{(4)}  +2 f' f^{(5)}  \big) f \right] 
= 0\,,
\ee
where
\bea
q_0\rovno 4\left[
9 \rn^4 \gamma \left(1-\Lambda  \rn^2\right) 
-6 \rn^2 (\alpha -3 \beta ) 
+3 \left(12 \eta _3+3 \eta _4+\eta _6+4 \eta _7\right)+2 \eta _8\right]\,,\\
q_1\rovno 3 \rn^4 \left[ 2 \rn^2 (\alpha -3 \beta ) -36 \eta _3+3 \eta _4+\eta _6-2(6 \eta _7 +\eta _8)  \right] 
\,,\\
q_2\rovno \rn^6 \left[3(12 \eta _3+3 \eta _4+ \eta _6+4 \eta _7)+2 \eta _8\right] \,,\\
q_3\rovno - 6  \rn^6 \left(\eta_1 - 3 \eta_2 \right)
\,.
\eea
Using \eqref{six_power_f}, we get the expansion of Eq. \eqref{eq_rr44},
\bea
&&q_0+q_1\sum_{l=2n-4}^\infty \Delta^l
\sum_{i=0}^{l-2n+4}a_i a_{l-2n+4-i}(n+i)
(n+i-1)(l-n-i+4)(l-3n-3i+7)\nonumber\\
&&+\sum_{l=3n-6}^\infty\Delta^l
\sum_{j=0}^{l-3n+6}\sum_{i=0}^j
a_ia_{j-i}a_{l-j-3n+6}
(n+i)(n+i-1)(n+j-i)\nonumber\\
&& \times \left\{ q_2 (n+j-i-1)(l-2n-j+6)
(l-5n-j-3i+11)\right.\nonumber\\
&& \left. \ \ +q_3(n+i-2)[(n+j-i-1)(n+j-i-2)
+(n+i-3)(4l-8n-6j+4i+18)]\right\} =0\,.
\label{Errsum}
\eea

Comparing the leading powers, we obtain three possibilities:
\begin{itemize}
    \item 
 1) $n>2$: Although the first term is initially leading, 
 Eq. \eqref{Errsum} implies that $q_0=0$. As a result, the second term (proportional to $q_1$) becomes the leading term and \eqref{Errsum} yields the following cases:
\begin{itemize}
    \item[$\ast$] Case 1a)
$n=3$.
\item[$\ast$]
Case 1b) 
$q_1=0$: In this special case, where $q_1$ vanishes, the last term in Eq. \eqref{Errsum} becomes the leading one and \eqref{Errsum} implies
\be
\label{cond_1a}
q_2 n(n-1)(2n-5)-q_3 (n-2)(5n^2-21n+20)=0\,.
\ee
\end{itemize} 
\item
 2) $n=2$: In this case, Eq. \eqref{Errsum} implies
\be
q_0+4q_1 a_0^2+8q_2 a_0^3=0.
\ee
\item
 3) $n<2$: In this case,  three subcases follow from \eqref{Errsum}:
 \begin{itemize}
     \item[$\ast$] 3a) $n=1$.
     \item[$\ast$] 3b) $n=0$.
\item[$\ast$] 
3c) 
 $q_2n(n-1)
(5-2n)+q_3(n-2)(5n^2-21n+20)=0$.
\end{itemize}
\end{itemize}

Let us now focus on the field equation $E_{\theta\theta}$.
It can be written as
\bea
&&\sigma_0+\sigma_1 f''+\sigma_2 f''^2
+\sigma_3\big( f' f^{(3)} + f f^{(4)} \big)
+\sigma_4 f''^3
+\sigma_5 f f^{(3)2}
+\sigma_6 \left[  f f^{(3)2} + \big( f' f^{(3)}  + f f^{(4)}  \big) f'' \right]
\nonumber\\
&&
+\sigma_7 \left[ 2 f'^2 f^{(4)} + \big(  f'' f^{(3)}  + 4 f  f^{(5)} \big) f' + \big( 3  f''  f^{(4)}  + f f^{(6)} \big) f \right] = 0,
\eea
where
\bea
\sigma_0\rovno
-8\left[9 \rn^6 \gamma  \Lambda  -6 \rn^2 (\alpha - 3 \beta ) +2 \left(36 \eta _3+9 \eta _4+3 \eta _6+12 \eta _7+2 \eta _8\right)\right]\,,\\
\sigma_1\rovno
-12 \rn^2 \left[ 3 \rn^4 \gamma  -36 \eta _3+3 \eta _4+\eta _6-12 \eta _7 -2\eta _8  \right]
\,,\\
\sigma_2\rovno
-12 \rn^6 (\alpha -3 \beta )\,,\\
\sigma_3\rovno
12 \rn^4 \left[ 2 \rn^2 (\alpha +6 \beta ) -4 (\eta _1-3\eta _2)+72 \eta _3-6 \eta _4+\eta _6-2 \eta _8 \right] \,,\label{sig3}\\
\sigma_4\rovno
-\rn^6 \left(36 \eta _3+9 \eta _4+3 \eta _6+12 \eta _7+2 \eta _8\right)
=-q_2\,,\\
\sigma_5\rovno
36  \rn^6 \left(\eta _1+3 \eta _2\right) \,,\\
\sigma_6\rovno
-6  \rn^6 \left(72 \eta _3+6 \eta _4-\eta _6-2 \eta _8\right) \,,\\
\sigma_7\rovno
24 \rn^6 \left(\eta _1+6 \eta _2\right)\,.
\eea
Inserting the series expansion for $f$ gives
\bea
&&\sigma_0+\sigma_1 \sum_{l=n-2}^{\infty}
\Delta^l a_{l-n+2}(l+2)(l+1)\nonumber\\
&&+\sum_{l=2n-4}^\infty\Delta^l
\sum_{i=0}^{l-2n+4}a_i a_{l-2n-i+4}
(n+i)(n+i-1)\big[
\sigma_2 (l-n-i+4)(l-n-i+3)\nonumber\\
&&\qquad \qquad \qquad\qquad +\sigma_3 (n+i-2)\left(
l+1\right)\big]\nonumber\\
&& +\sum_{l=3n-6}^\infty \Delta^l
\sum_{j=0}^{l-3n+6}\sum_{i=0}^j
a_i a_{j-i}a_{l-j-3n+6}
(n+i)(n+i-1)\nonumber\\
&&\qquad \times
\Big[
(n+j-i)(n+j-i-1)[\sigma_4(l-2n-j+6)
(l-2n-j+5)\nonumber\\
&&\qquad \qquad \qquad\qquad\qquad\qquad +
(\sigma_5+\sigma_6)(n+i-2)(n+j-i-2)]
\nonumber\\
&&\qquad 
+(n+i-2)(n+j-i)(n+j-i-1)
[\sigma_6(l-n+i-j+3)
+\sigma_7(l+n+3i-j-3)]
\nonumber\\
&&\qquad 
+\sigma_7 (n+i-2)(n+i-3)[
(n+j-i)(2l-2j+4i-4)+
(n+i-4)(n+i-5)]
\Big]=0.
\label{Eththsum}
\eea

Again, from the leading powers in the expansion, there are three possibilities:
\begin{itemize}
    \item 
i) $n>2$: the first term is leading, however, Eq. \eqref{Eththsum} implies 
$\sigma_0=0$ and thus the second becomes dominant. Then  \eqref{Eththsum}  gives $\sigma_1=0$. Finally, taking into account the third term, we get
\be
\sigma_0=0=\sigma_1, \ \ n (n-1)\sigma_2+(n-2)(2n-3)\sigma_3=0\, .
\label{cond_i}
\ee
\item ii) In the $n=2$ case, \eqref{Eththsum} gives
\be
\sigma_0+2\sigma_1 a_0+4\sigma_2 a_0^2+8\sigma_4 a_0^3=0.
\ee
\item
 iii) Finally, for  $n<2$, the last term is leading, and \eqref{Eththsum} gives three subcases
 \begin{itemize}
     \item[$\ast$] 
iiia) $n=1$,
\item[$\ast$] 
iiib) $n=0$, 
\item[$\ast$] 
iiic) 
\be
n(n-1)\{ n(n-1) \sigma_4
+(n-2)[(n-2) \sigma_5
+\sigma_6  (3n-5)]\}
+\sigma_7 (n-2)^3 (11 n -15)
=0\,. 
\ee
\end{itemize}
\end{itemize}

Combining these results we conclude that there are six classes of solutions compatible with
\eqref{Errsum} and \eqref{Eththsum}:
\begin{itemize}
    \item  
1a)
 $n=3$, $q_0=\sigma_0=\sigma_1=0$, $\sigma_3=-2\sigma_2$.
 \item  1b) 
$n>2$, $q_0=q_1=\sigma_0=\sigma_1=0$, also  \eqref{cond_1a}, \eqref{cond_i}, i.e.,
\begin{subequations}
\label{FinalCond1a}
\bea
n (n-1)\sigma_2+(n-2)(2n-3)\sigma_3\rovno 0\,,\\
q_2 n(n-1)(2n-5)-q_3 (n-2)(5n^2-21n+20)\rovno 0\,.
\eea
\end{subequations}
  \item  
 2)
 $n=2$ and
\begin{subequations}
\label{cond2IIii}
\bea
q_0+4q_1 a_0^2+8q_2 a_0^3&=& 0\,,
\label{cond2IIiia}\\
\sigma_0+2\sigma_1 a_0+4\sigma_2 a_0^2+8\sigma_4 a_0^3&= &0\,.
\eea
\end{subequations}
 \item   
3a)
$n=1$.
 \item   3b)
 $n=0$.
  \item  
 3c)
 $n<2$ and
 \begin{subequations}
 \label{FinaCond3c}
\bea
\label{483}
q_2n(n-1)
(2n-5)-q_3(n-2)(5n^2-21n+20)\rovno 0\,,\\
n(n-1)\{ n(n-1) \sigma_4
+(n-2)[(n-2) \sigma_5
+\sigma_6  (3n-5)]\}
+\sigma_7 (n-2)^3 (11 n -15)
\rovno 0\,. 
\eea
\end{subequations}
\end{itemize}

In this paper, we do not examine the higher orders of the expanded field equations, which might provide further constraints and branching among the above classes. 

In Appendix \ref{sec_con}, 
the connection between exact and power-series solutions of six-derivative gravity
is briefly discussed, with the purpose of showing that there exist solutions in all the above classes; the result is summarized in Table~\ref{table6der}. Last but not least, we also mention that more general solutions in the form of Puiseux series with fractional increments also exist, as the example in Sec.~\ref{SSec.Case5.6D} illustrates.

\vspace{2mm}
\begin{small}
\begin{table}[h]
		\begin{center}
			\begin{tabular}{|c|c|c|}
				\hline
Leading exponent & Conditions & Example in Sec.  \\[1mm] \hline\hline
$n=3$ & $q_0=\sigma_0=\sigma_1=0$, $\sigma_3=-2\sigma_2$ & \ref{SSec.Case3.6D} \\
$n>2$ & $q_0=q_1=\sigma_0=\sigma_1=0$ and  \eqref{FinalCond1a} & \ref{SSec.Case4.6D}\\
$n=2$ & \eqref{cond2IIii} & \ref{SSec6der-2}, \ref{SSec.Case3.6D}, \ref{SSec.Case4.6D} \\
$n<2$ & \eqref{FinaCond3c} & \ref{SSec.Case4.6D}\\
$n=1$ & none & \ref{sec_lin}, \ref{SSec6der-2}, \ref{SSec.Case3.6D} \\
$n=0$ & none & \ref{sec_lin}, \ref{SSec6der-2}, \ref{SSec.Case3.6D} \\
\hline
\end{tabular}\\[2mm]
\end{center}
	\caption{Classification of the closed-form solutions of Secs.~\ref{sec_lin}--\ref{SSec.Case4.6D} into the indicial classes of power series solutions.}
\label{table6der}
\end{table}
\end{small}

\subsection{(A)dS spacetime in six-derivative gravity}
\label{sec_six_AdS}

Having mentioned that the only Kundt spacetimes compatible with the metric~\eqref{SSexpandingmetric} are (A)dS and Minkowski spacetimes (see App.~\ref{AppC}), it is worth recalling that in six-derivative gravity, it can happen that (A)dS is a solution even without a cosmological constant in the action~\cite{Giacchini:2025gzw}. More generally, the model~\eqref{action-6der} admits the solution
\be
	\dd s^2= - \left( 1-\tfrac{\Lambda_{\text{eff}}}{3} r^2 \right) \dd t^2+ \left( 1-\tfrac{\Lambda_{\text{eff}}}{3} r^2 \right)^{-1} \dd  r^2+ r^2
     ( \dd \theta^2+\sin^2\theta\dd\phi^2),
	\ee
where the effective cosmological constant $\Lambda_{\text{eff}}$ satisfies
\be
\gamma (\Lambda_{\text{eff}}-\Lambda) - 16 \eta_3 \Lambda_{\text{eff}}^3 = 0.
\ee
Hence, if $\Lambda=0$ and $\gamma/\eta_3> 0$, then $\Lambda_{\text{eff}}=\pm \frac{1}{4} \sqrt{\gamma/\eta_3}$.


\section{Gravitational waves on static spherically symmetric Kundt background}
\label{sc:gravwaves}

For studying gravitational waves on static spherically symmetric Kundt backgrounds, let us use the static spherically symmetric Kundt metric in a form that is more suitable for studying gravitational waves.
Using the transformation  
\bea
t&=&\rn u-\rn \int\! \frac{\dd \rr}{\H(\rr)} \\
 r&=&\rn \rr\label{tr_conf_spher}
\eea
with
\be
f=-\H,\label{tr_f_H}
\ee
the metric \eqref{SSKundtmetric} transforms to
\be
	\dd s^2 = \rn^2
	\Big[ \dd \theta^2+\sin^2\theta\dd\phi^2 -2\,\dd u\,\dd \rr+\H(\rr)\,\dd u^2 \,\Big]\,.\label{metricconfK}
	\ee
 This metric will be used as a background for gravitational waves $h(u,\theta,\phi)$
\be
\dd s^2 \equiv \rn^2 \dd s^2_K=\rn^2
\Big[ \dd \theta^2+\sin^2\theta\dd\phi^2 -2\,\dd u\,\dd \rr+(\H(\rr)+h(u,\theta,\phi))\,\dd u^2 \,\Big]\,.
\label{QGwaves}
\ee    
Note that 
\be
{\Box_K h}=\BoxK h =  \frac{1}{\sin \theta} \partial_\theta\left(\sin \theta \partial_\theta h\right)+\frac{1}{\sin ^2 \theta } \partial_{\phi}^2 h
\ee
is in fact the Laplace--Beltrami operator on the unit 2-sphere $\mathrm{d} \theta^2+\sin ^2 \theta \mathrm{d} \phi^2$. In what follows, we analyze such exact gravitational-wave solutions in quadratic gravity as well as in six-derivative gravity with the solutions found in the previous sections serving as background. The resulting equations for $h$ become polynomial Laplace, Poisson, or heat equations on the 2-sphere, whose general regular solutions are discussed in App.~\ref{sc:Pboxeq}.


\subsection{Quadratic gravity }

Let us start with studying gravitational waves on static spherical Kundt backgrounds in quadratic gravity. We show that
\begin{lemma}
Let us consider vacuum quadratic gravity backgrounds \eqref{metricconfK}.
To have a nontrivial solution $h(u,\theta,\phi)$ (i.e., not $h=h(u)$)\footnote{{ \label{footnotehu} The function $h=h(u)$ can be transformed away via a suitable coordinate transformation.}} to the quadratic gravity field equations for metric \eqref{QGwaves}, the metric function $\H $  has to be at most  quadratic in $\rr$, i.e.,
\be
\H = h_2 \rr^2+h_1 \rr+h_0. \label{Hquadr}
\ee
Then the field equations reduce to (apart from a possible algebraic relation among $h_2$, $\Lambda$, and parameters of the particular theory)
\be 
(\alpha\BoxK+{\cal K} )\BoxK h=0\, \label{uu_wave_Gen} 
\ee
for some constant ${\cal K}$.\footnote{In the context of almost universal spacetimes, Eq. \eqref{uu_wave_Gen} was also derived in \cite{Kuchynkaetal19}; its non-local analogues appearing in infinite-derivative gravity were solved in \cite{Kolar:2021rfl}.}
\end{lemma}

\begin{proof}[Sketch of proof]

Assuming that either $h,_\theta\neq 0$ or $h,_\phi\neq 0$, the $u\theta$- and $u\phi$-components of the quadratic gravity field equations,
\bea
(2\alpha+3\beta)\H''' h,_\theta &=&0,\label{wave_02}\\
(2\alpha+3\beta)\H''' h,_\phi &=&0,\label{wave_03}
\eea
imply $(2\alpha+3\beta)\H'''=0$.  Thus, it immediately follows that $ \H'''=0  $ in the $(2\alpha+3\beta) \neq 0$ case.
Let us thus consider the $2\alpha+3\beta=0$ case. 

 The $\rr\rr$-component of the field equations implies  $(\alpha-3\beta) \H''''=0$. 
For $\alpha=3\beta$ together with $2\alpha+3\beta=0$, we get
$\alpha=\beta=0$, i.e., the quadratic gravity reduces to Einstein gravity.
From $\theta\theta$-component of the field equations,  we get $\H'''=0$.

 For $2\alpha+3\beta=0$ and $\alpha \neq 3\beta$, we immediately get $\H''''=0$. 
For the action to be nontrivial, either $\beta\neq 0$ or $\gamma\neq 0$. Then 
a linear combination of $u\rr$- and $\theta\theta$-components of the field equations implies $\H'''=0$. 
\end{proof}

Note that for backgrounds more general than \eqref{Hquadr}, the field equations of quadratic gravity impose overly restrictive constraints on $h$, leading typically to  $h=\mbox{const}$, which is incompatible with gravitational waves.

In general, the field equations for the metric \eqref{QGwaves} are coupled (see, e.g., \eqref{wave_02} and \eqref{wave_03}). However, when $\H'''=0$, all field equations except the $uu$-component decouple and no longer contain $h$. The $uu$-component reduces to \eqref{uu_wave_Gen} with
\be
 {\cal K}= \frac{4(\alpha-3\beta)-4 h_2(2\alpha+3\beta)-3\gamma \rn^2}{6 }.
\label{case_B2}    
    \ee
     Consequently, a quadratic $\H$ is not constrained by the presence of $h$ in the metric \eqref{QGwaves}, and thus gravitational waves can propagate in all backgrounds listed in Tab. \ref{Tab_constR}.

For special values of the parameters, ${\cal K}$ simplifies further.  Unless specified otherwise, the following cases correspond to Bachian-Nariai ($h_2>0$, $h_2\neq 1$), Bachian-Bertotti-Robinson ($h_2<0$, $h_2\neq -1$), or Pleba\'nski-Hacyan ($h_2=0$):
\begin{itemize}
    \item 
    $\gamma=0$  
    \begin{itemize}
        \item 
        $\alpha=3\beta$  with ${\cal K}=-2 h_2 \alpha$
\item 
$h_2=\pm 1$ ({Nariai/Bertotti-Robinson})
with  ${\cal K}=\frac{2[3\alpha-(h_2+1)(2\alpha+3\beta)]}{3 }$
    \end{itemize}
     \item 
     $h_2=2\Lambda \rn^2-1$.
\begin{itemize}
    \item 
    $\Lambda \rn^2=1$ (Nariai) with ${\cal K}=-\frac{4\alpha+24\beta+3\rn^2 \gamma}{6}$
\item 
$\gamma=0$, $\Lambda=0$ (Bertotti-Robinson) with ${\cal K}=2\alpha$
\item 
$\gamma=0$, $\alpha=3\beta$  with  ${\cal K}=-2 h_2\alpha$
\item 
$8\Lambda =\frac{3\gamma}{\alpha-3\beta}$ with ${\cal K}=\alpha\frac{4(\alpha-3\beta)-3\rn^2 \gamma}{2(\alpha-3\beta)}$
\end{itemize}
\end{itemize}

Note that these cases are not mutually exclusive (e.g., Nariai belongs to most of the cases above).
They are summarized in  Tab. \ref{tab_waveqg}.
 Einstein gravity appears as a limit $\alpha=0=\beta$.  Then the Einstein field equations reduce to $\BoxK h=0$, and the only vacuum background is Nariai ($h_2=1$, $\Lambda \rn^2=1$).
     Note that Pleba\'nski-Hacyan appears as a limiting case of
Bachian-Nariai or Bachian-Bertotti-Robinson with $h_2=0$.


\begin{table}[H]
	\begin{center}
		\begin{tabular}{|l|l|l|l|c|}
			\hline
			theory 
            & conditions & background& ${\cal K}$ \\[1mm] \hline\hline
$\alpha$, $\beta$, $\gamma$ arbitrary 
 &  $h_2=1 $, 
 $\Lambda \rn^2=1$  &
  Nariai 
  & $-\frac{4\alpha+24\beta+3\rn^2 \gamma}{6}$\\
     $\alpha\neq  3\beta$, $\gamma\neq 0$   
        & $h_2=2\Lambda \rn^2 -1$, $\Lambda=\frac{3\gamma}{8(\alpha-3\beta)}$ &
        Bachian-Nariai/Bachian-Bertotti-Robinson 
       &$\alpha \frac{4(\alpha-3\beta)-3\rn^2 \gamma}{2(\alpha-3\beta)}$\\
\hline
$\alpha$, $\beta$ arb., $\gamma=0$ 
& 
 $h_2=2\Lambda \rn^2 -1=-1$, $\Lambda=0$ &  
   Bertotti-Robinson  
   & $2\alpha$\\    
$\alpha$, $\beta$ arb., $\gamma=0$ 
&
   $h_2=\pm 1$  
   &  Nariai/Bertotti-Robinson
   & $\frac{2[3\alpha-(h_2+1)(2\alpha+3\beta)]}{3}$
\\
  \hline
$\alpha=3\beta$, $\gamma=0$ 
&
 $h_2=2\Lambda \rn^2 -1$& 
  Bachian-Nariai/Bachian-Bertotti-Robinson
 &  $-2 h_2\alpha$ \\
 $\alpha=3\beta$,  $\gamma=0$ 
& $h_2$ arbitrary 
& Bachian-Nariai/Bachian-Bertotti-Robinson
&  $-2 h_2\alpha$\\ 
$\alpha=3\beta$, $\gamma=0$ 
& $h_2=\pm 1$
&  Nariai/Bertotti-Robinson
& $\frac{2[3\alpha-(h_2+1)(2\alpha+3\beta)]}{3 }$
\\
\hline
$\alpha=0=\beta$, $\gamma\neq 0$ 
& $h_2=1$, $\Lambda \rn^2=1$ 
& Nariai in Einstein gravity & \\
\hline
            \end{tabular}\\[2mm]
			\caption{Vacuum
             Kundt backgrounds in quadratic gravity compatible with gravitational waves and the metric of the form \eqref{QGwaves}. In all cases, the metric function $\H$ has the form \eqref{Hquadr} and $h$ satisfies \eqref{uu_wave_Gen}. Pleba\'nski-Hacyan appears as a limiting case of
Bachian-Nariai or Bachian-Bertotti-Robinson with $h_2=0$.}
		\label{tab_waveqg}
            \end{center}
            \end{table}


Let us discuss solutions of \eqref{uu_wave_Gen} starting with the (globally) regular ones. As shown in Appendix~\ref{sc:Pboxeq} (in a more general case of polynomial Laplace equations), the solution can be written in terms of spherical harmonics ${Y_l^m}$,
\begin{equation}\label{eq:quadgrsolh}
    h=\frac{h_0^0(u)}{\sqrt{4\pi}}+\sum_{m=-{l_1}}^{m={l_1}}h_{l_1}^m(u) Y_{l_1}^m(\theta,\phi),
\end{equation}
where the coefficients $h_l^m (u)$ are arbitrary complex functions chosen such that $h$ is real. The two terms in \eqref{eq:quadgrsolh} correspond to the separate regular solutions of $\BoxK h=0$ and $(\BoxK+\TK) h=0$, where $\TK=\KK/\alpha$ for $\alpha\not=0$\footnote{For $\alpha=0$ and $\KK\not=0$, the equation reduces to $\BoxK h=0$.}. The modes $h_{l_1}^m(u)$ can be nonzero only if ${\TK = l_1 (l_1 + 1)}$ for some integer ${l_1 \in \mathbb{N}}$; otherwise, these modes vanish. As an example, here we list a few axisymmetric solutions (i.e., with ${m=0}$),
\begin{equation}
    h=\frac{h_0^0(u)}{\sqrt{4\pi}}+
    \begin{cases} 
     \sqrt{\frac{3}{4\pi }} h_{1}^0(u) \cos (\theta ), &\mathcal{ \tilde K}=2,
    \\
    \frac{1}{2} \sqrt{\frac{5}{4\pi }} h_{2}^0(u) \left(3 \cos ^2(\theta )-1\right), &\TK=6,
    \\
    \frac{1}{2} \sqrt{\frac{7}{4\pi }} h_{3}^0(u)\left(5 \cos ^3(\theta )-3 \cos (\theta )\right), &\TK =12.
    \end{cases}
\end{equation}
In the case of these `resonances', which are given by a constraint on a combination of coupling constants and $\rn$, quadratic gravity admits extra nontrivial solutions corresponding to the massive\footnote{Here the term ``massive'' is used in the usual higher-derivative-gravity terminology for the additional branch satisfying
	$(\triangle+\widetilde{\mathcal K})h=0$, with
	$\widetilde{\mathcal K}\ne0$, alongside the harmonic branch
	$\triangle h=0$.} gravitational waves (also in more general backgrounds), while the vacuum general relativity only admits the trivial regular solution of the massless equation (and only on the Nariai background); see \cite{Ortaggio02}, which focuses on impulsive waves.

Alongside the regular solutions, one may also construct a variety of singular solutions. Focusing on the solutions that correspond to two distributional $\delta$-sources localized at the poles ${\theta=0}$ and ${\theta=\pi}$, one can first notice that the even-parity source (with respect to ${\theta\to\pi-\theta}$) is inconsistent due to the root ${\Delta=0}$. This root introduces a zero mode, so the operator is not invertible, and the equation requires a compatibility condition on the source. Nevertheless, we can construct solutions for the odd-parity source satisfying
\begin{equation}\label{eq:oddsource}
   (\BoxK+\TK)\BoxK h=\mathcal{J}_{-}, \quad \mathcal{J}_{-}=J(u)\big[\delta(\cos\theta-1)-\delta(\cos\theta+1)\big].
\end{equation}
As shown in Appendix~\ref{sc:Pboxeq} (for general polynomial Laplace equations), the solution takes the form of a sum corresponding to the roots ${\Delta=0}$ and ${\Delta=-\TK}$,
\begin{equation}\label{eq:GF}
    h =\frac{1}{\TK}\left(h_{(0)}-h_{(\KK)}\right),
\end{equation}
provided $\TK\not=0$, ${h_{(0)}}$ and $h_{(\KK)}$ exist as solutions of ${\BoxK h_{(0)}=\mathcal{J}_{-}}$ and ${(\BoxK+\TK) h_{(\KK)}=\mathcal{J}_{-}}$, respectively. The former can be written using the odd-degree Legendre polynomials $P_{\ell_k}$, ${\ell_k=2k+1}$, 
\begin{equation}
    h_{(0)}=\frac{J(u)}{2}\sum_{k=0}^{\infty}\frac{2\ell_k+1}{-\ell_k(\ell_k+1)} P_{\ell_k}(\cos \theta)=\frac{J(u)}{2}\log\left(\tan\tfrac{\theta}{2}\right).
\end{equation}
The latter is only consistent with the source if ${\TK\neq \ell_{k}(\ell_{k}+1)}$, ${k\in\mathbb{N}_0}$. For example, if ${\TK=6}$, then
\begin{equation}
    h_{(6)}=\frac{J(u)}{2}\sum_{k=0}^{\infty}\frac{2\ell_k+1}{-\ell_k(\ell_k+1)+6} P_{\ell_k}(\cos \theta) = \frac{J(u)}{4}\left[-\left(3 \cos ^2\theta -1\right) \log \left(\cot \tfrac{\theta }{2}\right)+3 \cos \theta\right].
\end{equation}
If ${\TK=\nu(\nu+1)}$ with ${\nu\notin\mathbb{Z}}$, then
\begin{equation}
\begin{aligned}
    h_{(\KK)} &=\frac{J(u)}{2}\sum_{k=0}^{\infty}\frac{2\ell_k+1}{\nu(\nu+1)-\ell_k(\ell_k+1)}\,P_{\ell_k}(\cos\theta)
    \\
    &=\frac{J(u)}{2}\left[-\big(Q_\nu(\cos\theta)- Q_\nu(-\cos\theta)\big)+\frac{\pi \cos(\pi\nu)}{2\sin(\pi\nu)}\big(P_\nu(\cos\theta)- P_\nu(-\cos\theta)\big)\right],
\end{aligned}
\end{equation}
where $P_{\nu}$ and $Q_{\nu}$ are Legendre functions of the first and second kind. Note that the solution \eqref{eq:GF} can be used as a Green’s function to construct inhomogeneous solutions for an odd-parity source by convolution. Naturally, ignoring the (global) distributional character of the source, other local singular solutions exist; for example $h_{(0)}$ and $h_{(6)}$ separately, as well as ${h_{(2)}=1-\cos\theta\,\log\left(\cot\tfrac{\theta}{2}\right)}$, satisfy the equation \eqref{uu_wave_Gen} (with ${\TK=0,2,6}$) outside the singularities. In general relativity, the singular solutions such as ${h_{(0)}}$ are the only non-trivial waves on the Nariai spacetime. The singularities are typically interpreted as point sources or topological defects. In our case, the constants in the solution of \cite{Ortaggio02} become $u$-dependent. In fact, such a solution in different coordinates is known as the Khlebnikov-Ghanam-Thompson spacetime \cite{Khlebnikov86}, which was conjectured to be universal \cite{HerPraPra14} (see also \cite{HerPraPra17}).

\subsection{Six-derivative gravity}

The situation in six-derivative gravity is more complicated than in quadratic gravity. Depending on the background and the relations between the couplings of the model, the field equations can impose different constraints on $h(u,\theta,\phi)$. In what follows we give two examples.

\subsubsection*{Example 1}

Let us start with the backgrounds of Sec.~\ref{SSec6der-2} (i.e., Bertotti-Robinson, Nariai and their Bachian generalizations), for which $\H=h_2 \rr^2$ (with $h_2=-a_2 \rn^2$). In this case, only the $uu$-component is nontrivial, namely, 
\be
\label{uu-wave-SDG2}
\begin{aligned}
&
\big\lbrace
12 \eta _1 \BoxK^2 
+ \big[ 12  R_0^2 \alpha + 24 \eta _1 \left(h_2+3\right)-12 \eta _4 \left(h_2+1\right)-9\eta _5 \left(h_2-1\right)+2 \eta _6 \left(h_2+1\right)
\\
&
-24 \eta _7 \left(h_2+1\right)+12 \eta _8 \left(h_2+1\right)  \big] \BoxK 
 -6 R_0^4 \gamma - 8 R_0^2 [ \alpha  \left(2 h_2-1\right)+3 \beta  \left(h_2+1\right) ]
\\
& +24 \eta _1 \left[h_2 \left(h_2+4\right)+5\right]-72 \eta _3 \left(h_2+1\right)^2+6 \eta _4 \left(h_2-1\right) \left(h_2+3\right)+8 \eta _6 \left(h_2-1\right) h_2
\\
&
+24 \eta _7 \left(h_2^2-1\right)+8 \eta _8 \left(h_2+1\right) \left(h_2+4\right) 
\big\rbrace \BoxK h = 0 ,
\end{aligned}
\ee
which can be cast in the form
\be
(\BoxK +{\cal K}_1 ) ( \BoxK +{\cal K}_2 ) \BoxK h =0 \,,\ \ {\cal K}_1,{\cal K}_2=\mbox{const.},
\ee
if $\eta_1\neq0$.
Again, as follows from App.~\ref{sc:Pboxeq}, the general solution of this equation reads
\begin{equation}
    h=\frac{h_0^0(u)}{\sqrt{4\pi}}+\sum_{m=-{l_1}}^{m={l_1}}h_{l_1}^m(u) Y_{l_1}^m(\theta,\phi)+\sum_{m=-{l_2}}^{m={l_2}}h_{l_2}^m(u) Y_{l_2}^m(\theta,\phi),
\end{equation}
where $h_l^m$ are arbitrary complex functions chosen such that $h$ is real. The modes $h_{l_1}^m(u)$ and $h_{l_2}^m(u)$ are nonzero only if ${\KK_1 = l_1 (l_1 + 1)}$ and ${\KK_2 = l_2 (l_2 + 1)}$ for some ${l_1,l_2 \in \mathbb{N}}$, respectively; otherwise the sums are absent. If $l_1 = l_2$, the two sums coincide. Following the discussion after \eqref{eq:quadgrsolh} and \eqref{eq:oddsource}, it is then straightforward to construct specific examples of axisymmetric (${m=0}$) regular and singular solutions.

Notice that, for some specific backgrounds, certain couplings do not affect the perturbation. For instance, if $h_2=-1$ (Bertotti-Robinson) then~\eqref{uu-wave-SDG2} does not depend on $\eta_{3,7,8}$, whereas for $h_2=1$ (Nariai) it does not depend on $\eta_5$. It is also worth mentioning that although such backgrounds are generally independent of $\eta_5$ (see the discussion involving~\eqref{6der-eta5}), this coupling can affect the perturbations.

\subsubsection*{Example 2}

As a second example, consider the models defined by~\eqref{condF3SDGb} and~\eqref{condF3SDG}, that admit the background solution~\eqref{Sol.3.6D}, i.e.,
\be
\H(\rr) = h_3 \rr^3 - \frac{\rr}{3 h_3} , \qquad R_0^2 = 2 \sqrt{\frac{2 \eta _1-3 \eta _4}{\gamma}}.
\ee
The non-trivial components of the field equations for the perturbation $h(u,\theta,\phi)$ yield
\begin{subequations}\label{systemGW6der}
\begin{eqnarray}
E_{uu}: \quad && 
3 \left(R_0^4 \gamma +12 \eta _4\right) \BoxK^3 h 
+ 6 \big[ 5 R_0^4 \gamma +4 \alpha  R_0^2+3 \left(8 \eta _4+\eta _5-8 \eta _7\right) +9 h_3 \rr \left(R_0^4 \gamma -\eta _5-8 \eta _7\right) \big] \BoxK^2 h 
\nonumber
\\
&&
+ \big[ 2 \left(27 R_0^4 \gamma +8 \alpha  R_0^2+108 \eta _4-144 \eta _7\right) +6 h_3 \rr \left(39 R_0^4 \gamma -16 \alpha  R_0^2+108 \eta _4-144 \eta _7\right) \big] \BoxK h 
\nonumber
\\
&&
+ 12 h_3 \partial_u  \big[ 3 \left(R_0^4 \gamma +12 \eta _4\right)  \BoxK h + R_0^2 \big(  4 \alpha -6 R_0^2 \gamma \big) h \big]
= 0,
\label{Euu6der1}
\\
E_{u\theta} : \quad  &&
9 \left(R_0^4 \gamma-\eta _5-8 \eta _7\right) \partial_\theta \BoxK h  + \left(39 R_0^4 \gamma -16 \alpha  R_0^2+108 \eta _4-144 \eta _7\right) \partial_\theta h = 0 , \label{Euth6der1}
\\
E_{u\phi} : \quad &&
9 \left(R_0^4 \gamma-\eta _5-8 \eta _7\right) \partial_\phi \BoxK h  + \left(39 R_0^4 \gamma -16 \alpha  R_0^2+108 \eta _4-144 \eta _7\right) \partial_\phi h = 0 .
\label{Euph6der1}
\end{eqnarray}
\end{subequations}
From \eqref{Euth6der1} and \eqref{Euph6der1}, we get
\be 
\pi_2 \BoxK h  + \pi_1 h = {\cal T}(u)\,,
\label{eq_uth_uph}
\ee
where ${\cal T}$ is an arbitrary function of $u$
and
\bea
\pi_1\rovno 39 R_0^4 \gamma -16 \alpha  R_0^2+108 \eta _4-144 \eta _7\,,\\
\pi_2 \rovno 9 \left(R_0^4 \gamma-\eta _5-8 \eta _7\right) \,.
\eea
Substituting this result back into~\eqref{Euu6der1} and taking into account that $\BoxK {\cal T}(u)=0$, the dependence on $\rr$ is canceled and that equation reduces to
\be
\omega_3 \BoxK^3 h 
+ \omega_2  \BoxK^2 h 
+\omega_1 \BoxK h 
+\partial_u (\omega_{1u}  \BoxK  +\omega_{0u})h=0\,.
\label{eq_uu_r0}
\ee
where we have defined
\bea
\omega_3\rovno 3 \left(R_0^4 \gamma +12 \eta _4\right)\,,\\
\omega_2 \rovno 6 \big[ 5 R_0^4 \gamma +4 \alpha  R_0^2+3 \left(8 \eta _4+\eta _5-8 \eta _7\right)  \big] \,,\\
\omega_1\rovno   2 \left(27 R_0^4 \gamma +8 \alpha  R_0^2+108 \eta _4-144 \eta _7\right)   \,,\\
\omega_{1u} \rovno 36 h_3  \left(R_0^4 \gamma +12 \eta _4\right)  \,,\\
\omega_{0u} \rovno 12 h_3 R_0^2 \big(  4 \alpha -6 R_0^2 \gamma \big) \,.
\eea
Note that, differently from the cases of quadratic gravity and the previous example [see Eq.~\eqref{uu-wave-SDG2}], this equation might contain terms $\partial_u h$. Only for the even more particular models with $\omega_{1u}=0=\omega_{0u}$ (which also yields $\omega_3=0$), it happens that the equation for $h$ only involves the Laplace operator, as~\eqref{eq_uu_r0} simplifies to ($\omega_2  \BoxK  +\omega_1 )\BoxK h = 0$. 

If $\pi_2 \neq  0$ we can substitute
\be
 \BoxK h = \frac{- \pi_1 h + {\cal T}(u)}{\pi_2 }\label{boxh}\,
\ee
into \eqref{eq_uu_r0} to obtain
\be
(k_1
+  k_2\partial_u )h
=-( k_3  
+k_4   \partial_u ){\cal T}
\,,\label{druh}
\ee
where
\bea 
k_1\rovno \pi_1(-   \pi_1^2  \omega_3
+  \pi_1 \pi_2 \omega_2  
- \pi_2^2 \omega_1) \,,\\
k_2\rovno  \pi_2^2  (-  \pi_1\omega_{1u}   + \pi_2\omega_{0u} ) \,,\\
 k_3\rovno \pi_1^2\omega_3
-\pi_1 \pi_2\omega_2 
+   \pi_2^2 \omega_1 =-\frac{k_1}{\pi_1}  \,,\\
k_4\rovno  \pi_2^2    \omega_{1u}\,.
\eea 
Moreover, if $k_1\neq  0 \neq  k_2$ then \eqref{druh}
implies  that $h$ is of the form
\be
\label{form35}
h=e^{-\frac{k_1 u}{k_2} }\left[ f(\theta,\phi)-
\frac{1}{k_2}
\int^u e^{\frac{k_1 s}{k_2} }[k_3 {\cal T}(s)+k_4 {\cal T}'(s)] \mbox{d} s\right]\,.
\ee
Substituting it in \eqref{boxh}, it follows
\bea
\pi_2\BoxK f+\pi_1 f \rovno K\,,\\
e^{\frac{k_1 }{k_2}u } {\cal T}(u)+ 
\frac{\pi_1}{k_2}
\int^u e^{\frac{k_1 }{k_2} s}[k_3 {\cal T}(s)+k_4 {\cal T}'(s)] \mbox{d} s\rovno K\,,\label{eqT}
\eea
where $K$ is a constant. Differentiation of \eqref{eqT}
leads to
\be 
\omega_{0u} \pi_2^3 {\cal T}'(u)=0.
\ee 
Therefore, since we assumed $\pi_2\neq 0$, there are two possibilities:
\begin{itemize}
\item
If $\omega_{0u}=0$,
\be
\label{gw6derex2a}
h=e^{-\frac{k_1 }{k_2}u } f(\theta,\phi)+\frac{{\cal T}(u)}{\pi_1}.
\ee 
Indeed, recall that $\BoxK \mathcal{T}(u) = 0$ for any function $\mathcal{T}(u)$ and notice that in this case all the terms in~\eqref{eq_uu_r0} involve $\BoxK h$ or its higher-order derivatives, which leaves the dependence on $u$ undetermined.
\item If ${\cal T}=T_0$ is a constant,
\be
\label{gw6derex2b}
h=e^{-\frac{k_1 }{k_2}u } f(\theta,\phi)+\frac{T_0}{\pi_1}.
\ee 
\end{itemize}
Both cases correspond to $K=0$,\footnote{In general, $K$ is related to the   integration constant in~\eqref{form35} and can be absorbed into the definition of $f$, leading to solutions~\eqref{gw6derex2a} and~\eqref{gw6derex2b}.} so that $f$ satisfies
\be 
\pi_2\BoxK f+\pi_1 f =0\,.
\ee 
A non-zero regular solution to this equation only exists if ${\pi_1/\pi_2=l_0(l_0+1)}$ for some ${l_0\in\mathbb{N}}$ (see App.~\ref{sc:Pboxeq}) and takes the form
\begin{equation}
    f=\sum_{m=-{l_0}}^{m={l_0}}f_{l_0}^m Y_{l_0}^m(\theta,\phi),
\end{equation}
where $f_{l_0}^m$ are arbitrary complex constants such that $f$ is real.

There are other possibilities if $\pi_2$, $k_1$, or $k_2$ vanishes, but it is beyond the scope of this paper to scrutinize all the possible cases. However, to close with a few more examples, if $k_1=k_2=0$ and
assuming $\partial_u h\neq 0$, it is straightforward to identify three simple cases:
\begin{itemize}
\item If $\pi_1=\pi_2=0$ then~\eqref{eq_uth_uph} yields $\mathcal{T}=0$, and the perturbation is fully determined by Eq.~\eqref{eq_uu_r0}. 
Note that this equation might involve the terms $\partial_u h$ and $\partial_u \BoxK h$. Assuming ${\omega_3\neq0}$ and ${\omega_{1u}\neq0}$, we can rewrite \eqref{eq_uu_r0} to the form
\begin{equation}
    \mathcal{C}(\BoxK+\KK_1)(\BoxK+\KK_2)\BoxK h=\partial_u (\BoxK+\mathcal{L})h.
\end{equation}
The general regular solution can be written in terms of spherical harmonics, see App.~\ref{sc:Pboxeq}. It takes the form
\begin{equation}
    h=\sum_{m=-l_0}^{l_0} h_{l_0}^m(u) Y_{l_0}^m(\theta, \phi) + \sum_{l=0, l(l+1)\neq \mathcal{L}}^{\infty} \sum_{m=-l}^l  c_l^m e^{\lambda_l u} Y_l^m(\theta, \phi), \quad     \lambda_l=\frac{\mathcal{C}(x+\KK_1)(x+\KK_2)x}{x+\mathcal{L}}\bigg|_{x=-l(l+1)},
\end{equation}
where $h_{l_0}^m$ are arbitrary complex functions and $c_{l}^m$ arbitrary complex constants such that $h$ is real, and ${l_0\in\mathbb{N}_0}$ is such that ${\mathcal{L}=l_0(l_0+1)}$ 
if
${\mathcal{L}}$ equals either $0$, $\KK_1$, or $\KK_2$. If no such $l_0$ exists, the first sum is absent.

\item If $\pi_1=0$ and $\omega_{0u} = 0$, then~\eqref{druh} yields $\mathcal{T}(u)= K e^{-\frac{k_3 u}{k_4}}$ for some constant $K$. However, the only option compatible with Eq.~\eqref{eq_uth_uph} is $K=0$, which makes the perturbation governed by $\BoxK h = 0$.
\item If $\omega_{0u} = 0$, then any function $h = h(u)$ solves~\eqref{systemGW6der}.
\end{itemize}

(In consonance with our treatment of the solutions in six-derivative gravity, our goal here is mainly to show that such models can have solutions that are quite different from those of quadratic gravity. 
For the more general backgrounds discussed in Secs.~\ref{SSec.Case4.6D} and~\ref{SSec.Case5.6D}, even more components of the field equations are non-trivial and the analysis of the possible branches of models admitting such gravitational wave solutions becomes more complicated.)

\section{Summary}\label{sc:summary} 

In this paper, we have investigated static, spherically symmetric Kundt spacetimes as vacuum solutions of quadratic gravity and six-derivative gravity. In addition, we have also studied gravitational waves propagating on these vacuum backgrounds.

As shown in Appendix~\ref{AppC}, purely geometric constraints on spacetimes supporting a Kundt congruence—without employing the field equations of any particular theory—imply that the only Kundt spacetimes compatible with the standard static spherically symmetric metric \eqref{SSexpandingmetric} are (A)dS and Minkowski. Therefore, we employ a metric ansatz \eqref{SSKundtmetric} suitable for spherically symmetric Kundt spacetimes. Note however, that this ansatz does not include (A)dS and Minkowski geometries.

In generic quadratic gravity (coupling constants obeying ${\alpha\neq3\beta}$), we have derived closed-form expressions for all vacuum static, spherically symmetric Kundt solutions. These solutions include static spacetimes that are well known in the context of general relativity, where they appear as vacuum or electrovacuum solutions: specifically, the Nariai (vacuum), Bertotti-Robinson (electrovacuum), and Pleba\'nski-Hacyan (electrovacuum) geometries. In addition, this class also contains generalizations of these spacetimes with a non-trivial Bach tensor - Bachian-Nariai and Bachian-Bertotti-Robinson spacetimes. All of the above solutions have a constant Ricci scalar, and they have already been briefly discussed in \cite{PraPraPodSva21}, where spherically symmetric solutions to quadratic gravity with constant $R$ were studied, primarily with a focus on black hole solutions. Furthermore, the ${\alpha\neq3\beta}$ class contains a closed-form solution with a non-constant $R$ and a curvature singularity. Conversely, within the  ${\alpha= 3\beta}$ class, we have found only particular closed-form solutions, and we have studied more general solutions using Frobenius analysis.

Six-derivative gravity admits eight additional coupling constants. Consequently, the number of possible special subcases becomes too large to allow for a systematic study of all solutions. Instead, we focus on selected models to illustrate the variety of closed-form solutions that arise. We also determine the possible indicial families of Frobenius solutions and point out that closed-form solutions exist for each indicial family. Finally, for all closed-form solutions, we analyze the corresponding curvature singularities and their accessibility to geodesic observers.

Having established the static, spherically symmetric Kundt backgrounds in these theories, we then analyzed gravitational waves of the form \eqref{QGwaves} propagating through them. In general relativity, radiative solutions on some of these backgrounds were studied in \cite{Ortaggio02,OrtPod02} for vacuum and electrovacuum. Importantly, the field equations reduce to a single generalized wave equation only for a subset of the backgrounds obtained above. For quadratic gravity, this occurs only when the background metric function $f(r)$ is at most quadratic in $r$, whereas six-derivative gravity allows for more general backgrounds. For backgrounds compatible with these gravitational waves, we derived generalized wave equations of the form
\be
    (\BoxK +{\cal K}_1 )  \BoxK h =0 \,, \ \ \ (\BoxK +{\cal K}_1 ) ( \BoxK +{\cal K}_2 ) \BoxK h =0 \,,\ \ {\cal K}_1,{\cal K}_2=\mbox{const}
\ee
for quadratic gravity and six-derivative gravity, respectively. In contrast to general relativity, where all non-trivial solutions are necessarily singular, with singularities interpreted as point sources or topological defects, quadratic and six-derivative gravity admit non-trivial regular solutions corresponding to massive gravitational waves; however, singular solutions sourced by odd-parity $\delta$-distributions can still be constructed.

\section*{Acknowledgments}

B.L.G. and I.K. acknowledge financial support by Primus grant PRIMUS/23/SCI/005 from Charles University and the support from the Charles University Research Center Grant No. UNCE24/SCI/016.

V.P. and A.P. were supported by the Institute of Mathematics of the Czech Academy of Sciences (RVO 67985840) and by research grant GA\v CR 25-15544S. V.P. is also grateful to M. Ortaggio for a helpful discussion regarding his works \cite{Ortaggio02,OrtPod02}. 


\appendix

\section{Polynomial equations with Laplace operator on 2-sphere}\label{sc:Pboxeq}

In this appendix, we construct general solutions of the polynomial Laplace, Poisson, and heat equations that were derived in Sec.~\ref{sc:gravwaves}. First, we consider the polynomial heat equation of the form
\begin{equation}\label{eq:pboxhQ}
    P(\BoxK)h= \partial_u Q(\BoxK)h,
\end{equation}
where $P$ and $Q$ are two polynomials. Its special case, ${Q=0}$, is the polynomial Laplace equation,
\begin{equation}\label{eq:pboxhzero}
    P(\BoxK)h=0,
\end{equation}
The solution of these equations can be written in terms of the spherical harmonics ${Y_l^m}$ (see, e.g. \cite{Kolar:2021rfl}), which form a complete set of eigenfunctions of the Laplace operator on a 2-sphere,
\begin{equation}
    \BoxK Y_l^m(\theta, \phi)=- l(l+1) Y_l^m(\theta, \phi).
\end{equation}
This implies that any regular function ${h=h(u,\theta,\phi)}$ can be expanded as follows
\begin{equation}\label{eq:hsumY}
    h(u, \theta, \phi)=\sum_{l=0}^{\infty} \sum_{m=-l}^l h_l^m(u) Y_l^m(\theta, \phi),
\end{equation}
with the inverse being
\begin{equation}
    h_l^m(u) = \int_0^{\pi} d\theta' \int_0^{2\pi} d\phi' \, \sin\theta' \, h(u, \theta', \phi') \, \bar{Y}_l^m(\theta', \phi').
\end{equation}
Here, we used the normalization of $Y_l^m$,
\begin{equation}
    Y_l^m(\theta, \phi)=\sqrt{\frac{2 l+1}{4 \pi} \frac{(l-m)!}{(l+m)!}} P_l^m(\cos \theta) e^{i m \phi},
\end{equation}
which satisfies the orthonormality relation
\begin{equation}
    \int_0^\pi d \theta^{\prime} \int_0^{2 \pi} d \phi^{\prime} \sin \theta^{\prime} Y_l^m\left(\theta^{\prime}, \phi^{\prime}\right) \bar{Y}_{l^{\prime}}^{m^{\prime}}\left(\theta^{\prime}, \phi^{\prime}\right)=\delta_{l l^{\prime}} \delta_{m m^{\prime}},
\end{equation}
where $P_l^m$ are the associated Legendre functions. Inserting this into \eqref{eq:pboxhQ}, we find
\begin{equation}
    \sum_{l=0}^{\infty} \sum_{m=-l}^l \left[ P(-l(l+1))h_l^m(u)-Q(-l(l+1))\frac{d}{du}h_l^m(u) \right] Y_l^m(\theta, \phi)=0.
\end{equation}
Since $Y_l^m$ form an orthonormal basis, this immediately implies
\begin{equation}
    P(-l(l+1))h_l^m(u)-Q(-l(l+1))\frac{d}{du}h_l^m(u)=0,  \quad\forall l,m.
\end{equation}
Assuming ${Q(-l(l+1))\neq 0}$, this equation is solved by
\begin{equation}
    h_l^m(u) = c_l^m e^{\lambda_l u},
\end{equation}
where $c_l^m$ are arbitrary integration constants and we introduced
\begin{equation}
    \lambda_l:=\frac{P(-l(l+1))}{Q(-l(l+1))}.
\end{equation}

Let us denote a set of $l$ corresponding to the roots of $P$ and $Q$ that coincide with the eigenvalues of $\BoxK$ by
\begin{equation}
    \Xi:=\left\{l \in \mathbb{N}_0 \mid P(-l(l+1))=Q(-l(l+1))=0\right\}
\end{equation}
and a set of $l$ with non-vanishing values of $Q$ by
\begin{equation}
    \Sigma := \left\{l \in \mathbb{N}_0 \mid Q(-l(l+1))\neq 0\right\}.
\end{equation}
The general solution of \eqref{eq:pboxhQ} is then given by
\begin{equation}
    h_{\textrm{heat}}(u, \theta, \phi)=\sum_{l\in\Xi} \sum_{m=-l}^l h_l^m(u) Y_l^m(\theta, \phi) + \sum_{l\in\Sigma} \sum_{m=-l}^l  c_l^m e^{\lambda_l u} Y_l^m(\theta, \phi),
\end{equation}
where $h_l^m(u)$ are arbitrary functions of $u$. If ${Q=0}$, then ${\Sigma=\emptyset}$,
so the general solution of \eqref{eq:pboxhzero} reduces to
\begin{equation}
    h_{\textrm{Laplace}}(u, \theta, \phi)=\sum_{l\in\Xi} \sum_{m=-l}^l h_l^m(u) Y_l^m(\theta, \phi).
\end{equation}

Consider now the polynomial Poisson equation,
\begin{equation}\label{eq:Pheqc}
   P(\BoxK)h=C,
\end{equation}
where ${C=C(u)\neq0}$ is a non-zero function. It corresponds to the ${l=0}$ mode, i.e., ${C=\sqrt{4\pi}CY_0^0 }$. Inserting expansions into \eqref{eq:Pheqc}, we get 
\begin{equation}
    \sum_{l=0}^{\infty} \sum_{m=-l}^l \left[ P(-l(l+1))h_l^m(u)-\delta^0_l\delta_0^m \sqrt{4\pi}C(u) \right] Y_l^m(\theta, \phi)=0.
\end{equation}
Again, the orthonormality of $Y_l^m$ implies
\begin{equation}
     P(-l(l+1))h_l^m(u)-\delta^0_l\delta_0^m \sqrt{4\pi}C(u)=0, \quad \forall l,m.
\end{equation}
Clearly, the solution does not exist for ${P(0)=0}$. If ${P(0)\neq0}$, then
\begin{equation}
    h_{\textrm{Poisson}}(u, \theta, \phi)=\frac{C(u)}{P(0)}+\sum_{l\in\Xi} \sum_{m=-l}^l h_l^m(u) Y_l^m(\theta, \phi).
\end{equation}
Notice that the presence of solutions is only sensitive to the structure of the sets $\Sigma$ and $\Xi$. When applied to the equations in the main text, this can be translated to the specific conditions on the coupling constants.

To this point, we discussed (globally) regular solutions. As the space of singular solutions is large and convoluted, we restrict ourselves to analyzing only singular solutions corresponding to distributional $\delta$-sources; hence, we avoid higher-order distributional sources and singularities that fail to define distributions. Due to the presence of inexcludable zero modes of $P(\BoxK)$ on the compact boundaryless 2-sphere, not every source $J$ admits a solution of ${P(\BoxK)h=J}$ for a given polynomial $P$. Two common ways to define Green's function are thus assuming either the antipodal ${(\delta\pm\delta)}$-source (see, e.g., \cite{Kolar:2021rfl}) or the ${(\delta-1/{\text{vol}})}$-source (see, e.g., \cite{Kolar:2022kgx}). We focus on the former and consider the equation for antipodally-symmetric Green's function
\begin{equation}
   P(\BoxK)h=J(u)\big[\delta(\cos\theta-1)+\sigma\delta(\cos\theta+1)\big], \quad \sigma=\pm1.
\end{equation}
Using the above formulas together with the standard representation of the ${(\delta\pm\delta)}$-distribution,
\begin{equation}
    \delta(\cos\theta-1)+\sigma\delta(\cos\theta+1)=\sum_{k=0}^{\infty} \sqrt{(2\ell_k+1)\pi}\, Y_{\ell_k}^0(\theta,\phi), \quad
\ell_k=\begin{cases}
    2k, &\sigma=+1,
    \\
    2k+1, &\sigma=-1,
\end{cases}
\end{equation}
one finds that the non-zero modes are
\begin{equation}
    h^0_{\ell_k}(u)=\frac{J(u)\sqrt{(2\ell_k+1)\pi}}{P(-\ell_k(\ell_k+1))},
\end{equation}
under the assumption that ${P(-\ell_k(\ell_k+1))\neq0}$ for all ${k\in\mathbb{N}_0}$; otherwise, the operator has a kernel and the source must be modified. If we assume the roots to be simple (can be generalized to higher multiplicity) and denote them by ${-\mu_i}$, i.e., ${P(-\mu_i)=0}$, ${i=1,\dots,N}$, then the antipodally-symmetric Green's function reads
\begin{equation}\label{eq:GFsol}
\begin{aligned}
    h_{\text{GF}}(u,\theta,\phi) &=\sum_{k=0}^{\infty}h^0_{\ell_k}(u)Y_{\ell_k}^0(\theta, \phi)=\frac{J(u)}{2}\sum_{k=0}^{\infty}\frac{2\ell_k+1}{P(-\ell_k(\ell_k+1))} P_{\ell_k}(\cos \theta)
    \\
    &=\frac{J(u)}{2}\sum_{i=1}^{N}\frac{1}{P'(-\mu_i)}\sum_{k=0}^{\infty}\frac{2\ell_k+1}{-\ell_k(\ell_k+1)+\mu_i} P_{\ell_k}(\cos \theta)
    \\
    &=\frac{ J(u)}{2}\sum_{i=1}^{N}\frac{-\big(Q_{\nu_i}(\cos\theta)+\sigma Q_{\nu_i}(-\cos\theta)\big)+\frac{\pi \cos(\pi\nu_i)}{2\sin(\pi{\nu_i})}\big(P_{\nu_i}(\cos\theta)+\sigma P_{\nu_i}(-\cos\theta)\big)}{P'(-\mu_i)},
\end{aligned}
\end{equation}
where we introduced ${\nu_i(\nu_i+1)=\mu_i}$, used the partial-fraction decomposition
\begin{equation}
    \frac{1}{P(x)}=\sum_{i=1}^{N}\frac{1}{P'(-\mu_i)}\frac{1}{x+\mu_i},
\end{equation}
and the Legendre function summation identity 
\begin{equation}
    \sum_{l=0}^{\infty}\frac{2l+1}{\nu(\nu+1)-l(l+1)}\,P_l(x)=-2\,Q_\nu(x)+\frac{\pi \cos(\pi\nu)}{\sin(\pi\nu)}\,P_\nu(x), \qquad |x|<1,
\end{equation}
which can be restricted to even/odd terms as
\begin{equation}
    \sum_{k=0}^{\infty}\frac{2\ell_k+1}{\nu(\nu+1)-\ell_k(\ell_k+1)}\,P_{\ell_k}(x)=-\big(Q_\nu(x)+\sigma Q_\nu(-x)\big)+\frac{\pi \cos(\pi\nu)}{2\sin(\pi\nu)}\big(P_\nu(x)+\sigma P_\nu(-x)\big), \quad |x|<1, \quad \ell_k=\begin{cases}
    2k,\\
    2k+1.
\end{cases}
\end{equation}
The last line in \eqref{eq:GFsol} holds only if ${\nu_i\not\in\mathbb{Z}}$. Here are a few examples of equations together with their solutions (omitting the normalization):
\begin{equation}
\begin{aligned}
    P =\BoxK, \sigma=-1:  & \qquad\sum_{k=0}^{\infty}\frac{2\ell_k+1}{-\,\ell_k(\ell_k+1)}\,P_{\ell_k}(\cos\theta)=\log\left(\tan\tfrac{\theta}{2}\right), \qquad \ell_k = 2k+1 ,
    \\
    P =\BoxK+2, \sigma=+1: &\qquad \sum_{k=0}^{\infty}\frac{2\ell_k+1}{-\ell_k(\ell_k+1)+2}P_{\ell_k}(\cos\theta)=1-\cos\theta\,\log\left(\cot\tfrac{\theta}{2}\right),\qquad \ell_k=2k,
    \\
    P =\BoxK+6, \sigma=-1: &\qquad \sum_{k=0}^{\infty}\frac{2\ell_k+1}{-\ell_k(\ell_k+1)+6}P_{\ell_k}(\cos\theta)=-\tfrac12\left(3 \cos ^2\theta -1\right) \log \left(\cot \tfrac{\theta }{2}\right)+\tfrac32 \cos \theta,\qquad \ell_k=2k+1.
\end{aligned}
\end{equation}
Notice that the roots $-\mu$ and signs $\sigma$ are chosen so that ${P(-\ell_k(\ell_k+1))\neq0}$.

\section{Connection between exact and power-series solutions of six-derivative gravity}
\label{sec_con}

In general, the power series solutions discussed in Sec. \ref{sec_six_power} admit more free parameters than the exact solutions listed in Sections \ref{sec_lin}--\ref{SSec.Case4.6D}.
However, examples of these closed-form solutions cover all possible classes 1a), \dots 3c) introduced in Sec. \ref{sec_six_power}.
Let us now discuss these exact solutions in the context of these classes. The conclusions are summarized in Table~\ref{table6der}.

\subsubsection{Solutions  $f=a_0 + a_1 r $ with $a_1 \neq  0$ in Sec.  \ref{sec_lin}}

The solution discussed in Sec. \ref{sec_lin} corresponds to the classes 3a) or 3b) (respectively, if $n=1$ or $n=0$). Moreover, the requirement that
only the coefficients $a_0, a_1$ are non-vanishing  leads
to $q_0=0=\sigma_0$. Different combinations of these two conditions give
\bea
3\rn^4\gamma-4\rn^2(\alpha-3\beta)+w\rovno 0\,,\\
9\rn^4\gamma (2\Lambda \rn^2-1)+w\rovno 0\,,
\eea
where
\be
w=36 \eta _3+9 \eta _4+3 \eta _6+12 \eta _7+2 \eta _8 \,,
\ee
which are equivalent to~\eqref{CaseF2six}.

\subsubsection{Solutions $f=a_0 + a_1 r + a_2 r^2 $ with $a_2 \neq  0$  in Sec.  \ref{SSec6der-2}}

The solution discussed in Sec. \ref{SSec6der-2} can be contained in the classes 2), 3a) or 3b) (respectively, for $n=2$, $n=1$, or $n=0$), depending on whether $a_0$ and $a_1$ vanish. The requirement that the solution is a second-order polynomial 
leads to
\be
\label{exIVC}
\begin{split}
8q_2 a_2^3+4q_1 a_2^2+q_0 =  0\,,\\
8\sigma_4 a_2^3+4\sigma_2 a_2^2+2\sigma_1 a_2+\sigma_0 = 0\,,
\end{split}
\ee
which are equivalent to \eqref{6der-2-sys}. 
Their combination gives
\be
4a_2^2(\sigma_2+q_1)+2a_2\sigma_1
+\sigma_0+q_0=0\,.
\ee

Also, note that the conditions~\eqref{exIVC} match~\eqref{cond2IIii} (upon redefining $a_2 \rightarrow a_0$), which are obtained in a different context of $n=2$ solutions.

\subsubsection{Solutions $f=a_0 + a_1 r + a_2 r^2 + a_3 r^3$ with $a_3 \neq  0$ in Sec.  \ref{SSec.Case3.6D}}

The solutions discussed in Sec. \ref{SSec.Case3.6D} may correspond to the classes 1a), 2), 3a) or 3b), as $n\in\lbrace0,1,2,3\rbrace$. In Sec. \ref{SSec.Case3.6D}, $f(r)$ is cubic in $r$. The truncation of the power series after the cubic term  
requires
 \be
2(\eta_1-3\eta_2)+w=0
\ee
which is the condition
\eqref{condF3SDGb} and can be also expressed as 
$ q_3=3q_2$.
Higher orders of the field equations give relations between $a_0,a_1,a_2,a_3$
\bea
q_0+4q_1 (a_2^2-3a_1 a_3)+4q_2 a_2(2 a_2^2-9 a_1 a_3)+36q_3 a_0 a_3^2\rovno 0\,,
\nonumber\\
\sigma_0+2 \sigma_1 a_2+4\sigma_2 a_2^2
+6\sigma_3 a_1 a_3+8\sigma_4 a_2^3
+36  a_0 a_3^2 (\sigma_5+\sigma_6)+12 a_1 a_2 a_3(\sigma_6+ \sigma_7)\rovno 0\,,\nonumber\\
a_3 [\sigma_1 +2  a_2 (2\sigma_2
+\sigma_3 )
+4  a_2^2 (3 \sigma_4+\sigma_6 +\sigma_7)
+6 a_1 a_3 ( \sigma_5+2\sigma_6
+\sigma_7 )]
\rovno 0\,.\nonumber
\eea

A particular solution of the above system is the solution 
\eqref{Sol.3.6D}, \eqref{condF3SDG}.

As another, simpler, example of model with a solution in the class 1a) with $n=3$, $q_0 = \sigma_0 = \sigma_1 = 0$ but $q_1\neq0$, we mention the solution
\begin{equation}
  f(r) = a_3 r^3-\frac{45 a_3^2 r^4}{64 \Lambda } , \quad \rn^2 = \frac{1}{\Lambda },
\end{equation}
in the model with
\begin{equation}
\label{Exn=3}
\begin{gathered}
    \alpha = 3 \beta -\frac{3 \gamma }{2 \Lambda } , \quad
    \eta_1= -15 \eta_2 = \frac{45 \gamma }{32 \Lambda ^2},  
    \quad\eta_4= \frac{48 \beta  \Lambda -13 \gamma }{48 \Lambda ^2}, 
    \\
    \quad\eta_6= -\frac{48 \beta  \Lambda +35 \gamma }{16 \Lambda ^2} , 
    \quad \eta_7= \frac{39 \gamma -48 \Lambda  (\beta +12 \eta_3 \Lambda )}{64 \Lambda ^2} ,
    \quad \eta_8= \frac{9 (16 \Lambda  (\beta +8 \eta_3 \Lambda )-13 \gamma )}{32 \Lambda ^2},
\end{gathered}
\end{equation}
for which
$q_1=-\frac{18\gamma}{\Lambda^4}$.

\subsubsection{Solutions  $f = a_2 r^2 + a_p r^p$ with $a_p \neq  0$ in Sec.  \ref{SSec.Case4.6D}}

Finally, the solutions discussed in Sec. \ref{SSec.Case4.6D} are examples of the classes 1b) with $n>2$ provided $a_2=0$ and  $p>2$, or 2) with $n=2$ if $p>2$  and $a_2\neq  0$, or 3c)  with $n<2$ for  $p<2$. 
Of course, in these cases one has to assume that $p\in\mathbb{Z}$ if $a_2\neq0$ so that the solutions of Sec.~\ref{SSec.Case4.6D} fit the form of the series expansion~\eqref{six_power_f}, with integer increments.

The solution~\eqref{SolEx4N} and the related model~\eqref{Exa2+an} are such that $q_0=q_1=q_2=q_3=\sigma_0=\sigma_1=\sigma_2=\sigma_4=0$, which is enough to satisfy the conditions~\eqref{cond2IIii} of a solution in the class 2).

\section{Triviality of the standard static, spherically symmetric Kundt branch}

\label{AppC}

Let us now identify all Kundt spacetimes compatible with the
standard static, spherically symmetric metric
\be
	\dd s^2=-f(r) \dd t^2+\frac{1}{g(r)} \dd  r^2+ r^2
    ( \dd \theta^2+\sin^2\theta\dd\phi^2). \label{SSexpandingmetricApp}
	\ee
This is a Petrov type D or O metric.
For Petrov type D, the two multiple principal null directions are
\be
\ell^\pm = \frac{1}{f} \partial_t \pm \sqrt{\frac{g}{f}}\partial_r
\ee
and their expansion reads $ \frac{1}{2} \ell^a_{\ ;a}=\pm \sqrt{\frac{g}{f}}\frac{1}{r}$, which is strictly nonzero.

Let us assume that a spacetime described by \eqref{SSexpandingmetricApp} admits a Kundt congruence $\bl$. Such a null geodesic congruence is non-expanding, twist-free, and shear-free. In the Newman-Penrose (NP) formalism, this implies 
\be
\kappa = \sigma = \rho = 0.
\ee
This, using the NP equation (7.21b) of \cite{Stephanibook}, implies that the highest boost-weight component of the Weyl tensor vanishes,
\be
\Psi_0 = 0,
\ee
and thus the Kundt congruence is necessarily a principal null direction (PND) of \eqref{SSexpandingmetricApp}. If the metric \eqref{SSexpandingmetricApp} is of type D, it already possesses two expanding multiple PNDs and cannot admit another one. Thus, a Kundt spacetime compatible with \eqref{SSexpandingmetricApp} is necessarily of Petrov type O and all components of the Weyl tensor vanish.

Furthermore, the NP equations (7.21a) and (7.21k) of \cite{Stephanibook} imply 
\be
\Phi_{00} = \Phi_{01} = 0.
\ee
From the above equation, it follows that $\bl$ is an eigenvector of the Ricci tensor
\be
R^a_{\ b} \ell^b = \lambda \ell^a. \label{eqeigenR}
\ee
At the same time, for the metric \eqref{SSexpandingmetricApp}, $R^a_{\ b}$ is diagonal. 

Since $\bl$ is null, its time component $\ell^t$ is non-zero and $\lambda$ is fixed by
\be
\lambda = R^t_{\ t}=R^r_{\ r},
\ee
which implies $f/g=$const. (which can be set to $1$ by rescaling). The expanding PNDs of \eqref{SSexpandingmetricApp} have only $\partial_t$ and $\partial_r$ nonzero components. The Kundt direction $\bl$ cannot be proportional to these PNDs (since it is non-expanding) and thus either $\ell^\theta$ or $\ell^\phi$ is non-vanishing. From \eqref{eqeigenR} it then follows that
\be
\lambda = R^\phi_{\ \phi}=R^\theta_{\ \theta}.
\ee
Therefore, 
\be
R^a_{\ b}=\lambda \delta^a_{\ b}.
\ee
The contracted Bianchi identities then imply $\lambda =$const. and thus a Kundt spacetime compatible with \eqref{SSexpandingmetricApp} is necessarily a maximally symmetric space and therefore (A)dS or Minkowski.



\providecommand{\href}[2]{#2}\begingroup\raggedright\endgroup

\end{document}